\documentclass[onecolumn,nofootinbib,thightenlines,notitlepage,tightenlines,longbibliography,superscriptaddress,11pt]{revtex4-1} 
\newcommand{\pagenumbaa}{1}
\usepackage{graphicx}
\usepackage{amssymb}
\usepackage{amstext}

\usepackage[usenames,dvipsnames]{xcolor}

\usepackage{enumerate}

\usepackage{tikz}
\usetikzlibrary{chains}
\usetikzlibrary{fit}
\usepackage{pgflibraryarrows}		
\usepackage{pgflibrarysnakes}		

\usepackage{epsfig}
\usetikzlibrary{shapes.symbols,patterns} 
\usepackage{pgfplots}

\usepackage[caption=false]{subfig}

\usepackage{graphicx}

\usepackage{amsthm}
\usepackage{amsmath}
\usepackage{amsfonts}
\usepackage{amssymb,amstext}
\usepackage{bbm} 
\usepackage{dsfont}

\usepackage[colorlinks=true,urlcolor=blue, hyperindex,breaklinks=true] {hyperref}
\usepackage{todonotes}

\usepackage{footnote}

\theoremstyle{plain}
\newtheorem{mythm}{Theorem}[section]
\newtheorem{myprop}[mythm]{Proposition}
\newtheorem{mycor}[mythm]{Corollary}
\newtheorem{mylem}[mythm]{Lemma}

\theoremstyle{definition}

\newtheorem{myex}[mythm]{Example}
\newtheorem{myremark}[mythm]{Remark}
\newtheorem{myass}[mythm]{Assumption}

\newcommand{\bracket}[1]{[#1]}

\newcommand{\UB}{\mathrm{UB}} 
\newcommand{\LB}{\mathrm{LB}} 

\newcommand{\Ent}{\mathrm{Ent}} 

\newcommand{\W}{\mathsf{W}} 

\def\setC{\mathsf{c}}

\newcommand{\introsection}[1]{\emph{#1}.---}
\newcommand{\I}[2]{I\!\left({#1},{#2}\right)} 
\newcommand{\D}[2]{D\!\left({#1}\right| \!\!\left|{#2}\right)} 

\newcommand{\Prob}[1]{\,{\mathds P} \!\left[#1\right]} %
\newcommand{\E}[1]{\,{\mathds E}\!\left[#1\right]} 
\newcommand{\Var}[1]{\,{\mathsf{Var}}\!\left[#1\right]} 

\DeclareMathOperator{\st}{s.t.}

\newcommand{\transp}{\ensuremath{^{\scriptscriptstyle{\top}}}}

\newcommand{\Wn}{\mathsf{W}^{(V,n)}}

\newcommand{\Borelsigalg}[1]{\ensuremath{\mathcal{B}\!\left(#1\right)}}

\newcommand{\R}{\ensuremath{\mathbb{R}}}
\newcommand{\Rp}{\ensuremath{\R_{\geq 0}}}
\newcommand{\Rsp}{\ensuremath{\R_{>0}}}

\newcommand{\Lp}[1]{\mathrm{L}^{#1}}

\newcommand{\brackett}[1]{ \left \lbrace \text{#1}\right \rbrace }

\newcommand{\ceil}[1]{\ensuremath{\lceil{#1}\rceil}}

\allowdisplaybreaks

\usepackage{makecell}

\long\def\symbolfootnote[#1]#2{\begingroup\def\thefootnote{\fnsymbol{footnote}}\footnote[#1]{#2}\endgroup}

\begin{document}

\symbolfootnote[0]{
The material in this paper was presented in part at the 52nd Annual Allerton Conference on Communication, Control, and Computing, Monticello, IL, October 2014. $\vspace{1mm}$}

\title{Capacity of Random Channels with Large Alphabets}

 \author{Tobias Sutter}
 \email[]{$\brackett{sutter,\,lygeros}$@control.ee.ethz.ch} 
 \affiliation{Automatic Control Laboratory, ETH Zurich, Switzerland}
 
 \author{David Sutter}
 \email[]{suttedav@phys.ethz.ch}
 \affiliation{Institute for Theoretical Physics, ETH Zurich, Switzerland}

 \author{John Lygeros}
 \email[]{$\brackett{sutter,\,lygeros}$@control.ee.ethz.ch} 
 \affiliation{Automatic Control Laboratory, ETH Zurich, Switzerland}


\begin{abstract}

We consider discrete memoryless channels with input alphabet size $n$ and output alphabet size $m$, where $m=\ceil{\gamma n}$ for some constant $\gamma>0$. The channel transition matrix consists of entries that, before being normalized, are independent and identically distributed nonnegative random variables $V$ and such that $\E{(V \log V)^2}<\infty$. We prove that in the limit as $n\to \infty$ the capacity of such a channel converges to $\Ent(V) / \mathds{E}[V]$ almost surely and in $\Lp{2}$, where $\Ent(V):= \E{V\log V}-\E{V}\log \E{V}$ denotes the entropy of $V$. 
We further show that, under slightly different model assumptions, the capacity of these random channels converges to this asymptotic value exponentially in $n$. 
Finally, we present an application in the context of Bayesian optimal experiment design.

\end{abstract}

 \maketitle

 \setcounter{page}{\pagenumbaa}  
 \thispagestyle{plain}

\section{Introduction}

Since Shannon's seminal 1948 paper \cite{shannon48}, channel capacity has become a fundamental concept in information theory, specifying the asymptotic limit on the maximum rate at which information can be transmitted reliably over a channel. In this work, we restrict ourselves to discrete memoryless channels (DMCs), that comprise a finite input alphabet $\mathcal{X} = \{1,2,\hdots,n\}$, a finite output alphabet $\mathcal{Y} = \{1,2,\hdots,m\}$, and a conditional probability mass function expressing the probability of observing the output symbol $y$ given the input symbol $x$, denoted by $\W_{x,y}$. Any DMC can be represented by a stochastic matrix $\W=(\W_{x,y})_{x\in\mathcal{X},y\in\mathcal{Y}} \in [0,1]^{n\times m}$, whose rows are normalized, i.e., $\sum_{y\in \mathcal{Y}} \W_{x,y} =1$ for all $x\in \mathcal{X}$.
According to Shannon \cite{shannon48}, the channel capacity of a DMC $\W$ is given by
\begin{equation} \label{eq:capacity}
C(\W)=\max \limits_{p \in \Delta_n} \I{p}{\W}, 
\end{equation} 
where $\I{p}{\W}:=\sum_{x \in \mathcal{X}} p(x) \D{\W_{x,\cdot}}{(p\W)(\cdot)}$ denotes the mutual information and 
$\Delta_{n}:=\{  x\in\R^{n} | \sum_{i=1}^{n} x_{i}=1, x_{i}\geq 0 \text{ for all }i \}$ the $n$-simplex. The channel law is described by $\W_{x,y}=\Prob{Y=y|X=x}$, $(p\W)(\cdot)$ denotes the probability distribution of the channel output induced by $p$ and $\W$ which is given by $(p\W)(y):=\sum_{x \in \mathcal{X}} p(x) \W_{x,y}$ for $y \in \mathcal{Y}$ and $\D{\cdot}{\cdot}$ is the relative entropy.
The optimization problem \eqref{eq:capacity}, while being convex, in general does not admit a closed form solution. Therefore, channel capacities are usually approximated with numerical algorithms such as the Blahut-Arimoto algorithm \cite{blahut72,arimoto72}, whose computational complexity of finding an additive $\varepsilon$-close solution scales cubically in the alphabet size, and as such the computational cost required for an acceptable accuracy for channels with large input alphabets can be considerable (see \cite{ref:Sutter-15} for more details).

In this paper, we are interested in a particular class of DMCs which are characterized by the property that each entry of their channel matrix is an i.i.d.\ random variable before the rows are normalized.
Two different scenarios are considered; first we assume that each entry of the channel transition matrix is a nonnegative i.i.d.\ random variable $V$ before being normalized and that $m=\ceil{\gamma n}$ for some constant $\gamma>0$.
Using duality of convex optimization, we prove in Theorem~\ref{thm:main} that as $n \to \infty$ the capacity of such a (random) DMC converges to $\tfrac{\mu_2}{\mu_1}-\log \mu_1$ almost surely and in $\Lp{2}$, where $\mu_1:=\E{V}>0$ and $\mu_2:=\E{V \log V}$. 
Second, we consider a more general setup under slightly different model assumptions, where each entry $V_{x,y}$ of the channel transition matrix, before being normalized, is independent and distributed on the nonnegative real line such that for all $x\in\mathcal{X}$ and for all $y\in\mathcal{Y}$ we have $\mu_{1,n}:=\tfrac{1}{m} \sum_{y \in \mathcal{Y}} \E{V_{x,y}}=\tfrac{1}{n} \sum_{x \in \mathcal{X}} \E{V_{x,y}}$ and $\mu_{2,n}:=\tfrac{1}{m} \sum_{y \in \mathcal{Y}} \E{V_{x,y} \log V_{x,y}}$.
In Theorem~\ref{thm:finiteSize} we show that the capacity of such a random DMC converges exponentially in $n$ to its asymptotic value $\lim_{n\to\infty}\tfrac{\mu_{2,n}}{\mu_{1,n}}-\log \mu_{1,n}$ in probability. Therefore, for the considered class of random DMCs the capacity, as the alphabet sizes tend to infinity, admits a closed form expression. We will show that this favourable property can be exploited in applications in the context of Bayesian optimal experiment design.

In the literature there exists a variety of extensively studied channel models that are described by random constructions, where one observes that in the limit as the blocklength tends to infinity the capacity converges to a deterministic value. This is sometimes viewed as a manifestation of of diversity \cite{BPS98,TV14}. A common model studied in \cite{BPS98,TV14} is of the form $y=Gx +w$, where $x$ is an $n$-dimensional input vector and $y$ represents an $m$-dimensional output vector. $G$ is modeled as a random matrix (the simplest example is the one where $G$ has i.i.d.\ entries) and $w$ denotes additive noise.
To the best of our knowledge the random channel model that is considered in this article has never been addressed directly in the literature. 

Understanding the behavior of random channels is important from a theoretical viewpoint as random constructions can serve as a powerful tool in order to prove statements. For example in quantum information theory there was a long-standing conjecture that the \emph{Holevo capacity} of a quantum channel is additive~\cite{holevo_book}. A few years ago, Hastings showed that the conjecture is false~\cite{hastings09}, by constructing a random, high-dimensional channel whose Holevo capacity is not additive. Despite considerable effort, there is no deterministic, low-dimensional quantum channel known for which we can prove that the Holevo capacity is not additive. This example shows the power of random constructions as a proof technique.

%

\vspace{2mm}
\introsection{Notation}
The logarithm with basis 2 is denoted by $\log(\cdot)$ and the natural logarithm by $\ln(\cdot)$. We consider DMCs with an input alphabet $\mathcal{X}=\{ 1,2,\hdots,n \}=:[n]$ and an output alphabet $\mathcal{Y}=\{ 1,2,\hdots,m \}=:[m]$. The channel law is summarized in a stochastic matrix $\W\in\mathcal{M}_{n,m}$, where $\W_{x,y}:=\Prob{Y=y|X=x}$ and $\mathcal{M}_{n,m}$ denotes the set of all stochastic $n\times m$ matrices.
The input and output probability mass functions are denoted by the vectors $p\in \Delta_{n}$ and $q\in\Delta_{m}$, where we define the standard $n$-simplex as $\Delta_{n}:=\left\{  x\in\R^{n} | \sum_{i=1}^{n} x_{i}=1, x_{i}\geq 0 \text{ for all }i \right\}$. For a probability mass function $p \in \Delta_{n}$ we denote the Shannon entropy by $H(p):=-\sum_{i=1}^n p_i \log p_i$. It is convenient to introduce an additional variable for the conditional entropy of $Y$ given $X$ as $r\in\R^{n}$, where $r_{x}:=-\sum_{y=1}^{m}\W_{x,y}\log \W_{x,y}$.
We denote the maximum (resp.~minimum) between $a$ and $b$ by $a \vee b$ (resp. $a \wedge b$) and by $\ceil{\cdot}$ the ceiling function. Given a nonempty set $A\subset\R$, its Borel $\sigma$-algebra is denoted by $\Borelsigalg{A}$. The uniform distribution with support $A$ is denoted by $\mathcal{U}(A)$ and the exponential distribution with rate parameter $\lambda>0$ by $\mathcal{E}(\lambda)$. The Dirichlet distribution on the $n$-simplex with concentration parameter $\alpha\in\Rp^{n}$ is denoted by $\mathsf{Dir}(\alpha_{1}, \hdots, \alpha_{n})$ and the lognormal distribution with rate parameters $z\in\R$ and $\sigma>0$ by $\ln\mathcal{N}(z,\sigma)$. The Dirac delta distribution is denoted by $\delta(\cdot)$.
By convention when refering to sets or functions, \emph{measurable} means \emph{Borel-measurable}.
Let $U$ be a nonnegative real-valued integrable random variable. The entropy of $U$ is defined as
$\Ent(U):=\E{U\log U} - \E{U}\log\E{U}$.

\vspace{2mm}
\introsection{Structure} 
In Section~\ref{sec:main} the asymptotic capacity of random DMCs having the form explained above is determined. Section~\ref{sec:simulation} contains a numerical simulation of a random DMC whose rows are uniformly distributed over the $n$-simplex. An application of the asymptotic capacity in terms of optimal design of experiments is presented in Section~\ref{sec:application}. Section~\ref{sec:finiteSize} proves the exponential rate of convergence for the capacity of such random DMCs under slightly different model assumptions. 

\section{Asymptotic Capacity} \label{sec:main}
Consider a probability space $(\Omega, \mathcal{A}, \mathds{P})$ and let $(V_{x,y})_{x\in [n], y \in [m]}$ be a sequence of i.i.d.\ nonnegative random variables on $\Omega$. We define the channel transition matrix\footnote{In Assumption~\ref{ass:input:output:size} the output alphabet size is assumed to be a function of the input alphabet size and as such the index $m$ is suppressed in the notation of the channel matrix.} $\Wn:=(\Wn_{x,y})_{x\in [n],y\in[m]}$ by $\Wn_{x,y}=V_{x,y} / \sum_{y \in [m]}V_{x,y}$, that can be easily verified to be a stochastic matrix, i.e., $0\leq \Wn_{x,y}$ for all $x\in [n],y\in [m]$ and $\sum_{y\in [m]} \Wn_{x,y}=1$ for all $x\in [n]$. 
We impose the following assumption on the random variables $V_{x,y}$.
\begin{myass} \label{ass:finite:moments}
The random variables $V_{x,y}$ are such that $\E{V_{x,y}}>0$ and $\E{(V_{x,y} \log V_{x,y})^{2}}<\infty$.
\end{myass}
Note that Assumption~\ref{ass:finite:moments} implies that $\E{V_{x,y}^{2} }<\infty$. 
The following assumption provides a relation between the input and output alphabet size that is required for the main theorem.
\begin{myass}\label{ass:input:output:size}
There is a positive constant $\gamma\in\Rsp$ such that the output alphabet size is given by $m=\ceil{\gamma n}$.
\end{myass}
It can be easily shown that the capacity $C(\Wn)$ of such a (random) DMC as well as the optimal input distribution are random variables. In order to do so, note that the mapping $\Rp^{n\times m}\ni (V_{x,y})_{x\in [n],y\in [m]} \mapsto (\Wn_{x,y})_{x\in [n],y\in [m]} = V_{x,y}/\sum_{y \in [m]}V_{x,y}\in\mathcal{M}_{n,m}$ constructing the channel clearly is measurable and therefore, invoking Lemma~\ref{lem:measurability} (see Appendix~\ref{app:measurability} for details), the channel capacity $C(\Wn)$ is a function from $\Omega$ to $\Rp$ that is $(\mathcal{A}, \mathcal{B}(\Rp))$-measurable and hence a random variable. Therefore, we can state the main result as follows, where we define $\mu_1:=\E{V_{x,y}}$ and $\mu_2:=\E{V_{x,y}\log V_{x,y}}$.
\begin{mythm}[Asymptotic capacity] \label{thm:main}
Under Assumptions~\ref{ass:finite:moments} and \ref{ass:input:output:size}, as $n\to \infty$ the capacity $C(\Wn)$ converges to $\tfrac{\mu_2}{\mu_1} - \log \mu_1$ almost surely and in $\Lp{2}$. 
\end{mythm}
\begin{proof}
See Section~\ref{sec:pfmain}.
\end{proof}
Under weaker assumptions on the channel matrix we can prove a weaker convergence statement for the asymptotic capacity.
\begin{mycor}[Asymptotic capacity] \label{cor:main}
Under Assumption~\ref{ass:input:output:size} and $\E{V_{x,y}}>0$ and $\E{V_{x,y} \log V_{x,y}}<\infty$, the capacity $C(\Wn)$ of the DMC, as $n\to \infty$, converges to $\tfrac{\mu_2}{\mu_1} - \log \mu_1$ almost surely. 
\end{mycor}
\begin{proof}
Follows directly from the proof of Theorem~\ref{thm:main}.
\end{proof}

Let us discuss some implications of Theorem~\ref{thm:main} and provide a few examples.

\begin{myremark}[Connection to $\Phi$-entropy]\label{rmk:phi:entropy}
For any convex function $\Phi:\Rp\to\R$, the \emph{$\Phi$-entropy} of a nonnegative real-valued integrable random variable $U$ is defined by
\begin{equation*}
\Ent_{\Phi}(U):=\E{\Phi(U)} - \Phi(\!\E{U}),
\end{equation*}
see \cite[Chapter~14]{ref:boucheron-13} for a comprehensive study. Let us consider the function $\Phi(u)=u\log u$ and denote the resulting $\Phi$-entropy by $\Ent(U)$ that simplifies to
\begin{equation*} 
\Ent(U)=\E{U\log U} - \E{U}\log\E{U}.
\end{equation*}
Under Assumptions~\ref{ass:finite:moments} and \ref{ass:input:output:size}, using the $\Phi$-entropy, Theorem~\ref{thm:main} can be stated equivalently as
\begin{equation*}
\lim\limits_{n\to\infty} C(\Wn) = \frac{\Ent(V_{x,y})}{\E{V_{x,y}}}.
\end{equation*}
\end{myremark}

\begin{myremark}[Properties of the asymptotic capacity] \label{rmk:properties:asym:cap}
The asymptotic capacity described in Theorem~\ref{thm:main}
\begin{enumerate}[(i)]
\item is \emph{nonnegative} by Jensen's inequality, since $\Rp \ni xÊ\mapsto x \log x \in \R$ is a convex function. 
\item can be \emph{zero}. Consider random variables $V_{x,y}$ such that $\Prob{V_{x,y}=\alpha}=1$ for some $\alpha \in \Rsp$. This then gives $\mu_1=\E{V_{x,y}}=\alpha$ and $\mu_2=\E{V_{x,y} \log V_{x,y}}= \alpha \log \alpha$, which leads to $\tfrac{\mu_2}{\mu_1}-\log \mu_1 = 0$.
\item can be \emph{arbitrarily large}. Consider random variables $V_{x,y}$ such that for some $\varepsilon \in (0,1)$, $\Prob{V_{x,y}=0}=1-\varepsilon$ and  $\Prob{V_{x,y}=1}=\varepsilon$. This then gives $\mu_1=\E{V_{x,y}}=\varepsilon$ and $\mu_2=\E{V_{x,y} \log V_{x,y}}= 0$ and hence $\tfrac{\mu_2}{\mu_1}-\log \mu_1 = \log \tfrac{1}{\varepsilon}$ which tends to infinity as $\varepsilon \to 0$.
\item admits the \emph{homogeneity property} $\lim_{n\to\infty}C(\W^{(\alpha V,n)}) = \lim_{n\to\infty} C(\Wn)$ for any $\alpha>0$. This follows by Remark~\ref{rmk:phi:entropy}, as
\begin{equation*}
\lim_{n\to\infty}C(\mathsf{W}^{(\alpha V,n)}) = \frac{\Ent(\alpha V_{11})}{\E{\alpha V_{11}}} = \frac{\Ent(V_{11})}{\E{V_{11}}} = \lim_{n\to\infty} C(\Wn),
\end{equation*}
where the second equality uses \cite[Remark~3.3.1]{ref:Raginsky-13}
\begin{align*}
\Ent(\alpha V_{11}) 	&= \E{\alpha V_{11}\log(\alpha V_{11})} - \E{\alpha V_{11}}\log(\E{\alpha V_{11}}) \\
						&= \alpha \E{V_{11}\log V_{11}} - \alpha \E{V_{11}}\log \E{V_{11}} = \alpha \Ent(V_{11}). 
\end{align*}

\end{enumerate}
\end{myremark}


\begin{myex}[Exponential distribution] \label{ex:exp}
Consider a DMC as defined above using an exponential distribution with rate parameter $\lambda > 0$. Then for $n \to \infty$ its capacity converges to $\tfrac{1-\kappa}{\ln 2}$ almost surely and in $\Lp{2}$, where $\kappa$ denotes Euler's constant. This follows directly from Theorem~\ref{thm:main}, since for $V_{x,y}\sim\mathcal{E}(\lambda)$ we have $\mu_1=\E{V_{x,y}}=\tfrac{1}{\lambda}$ and $\mu_2=\E{V_{x,y}\log V_{x,y}}=\tfrac{1-\kappa-\ln \lambda}{\lambda \ln 2}$. The fact that the asymptotic capacity is constant (i.e., independent of $\lambda$) is a direct consequence of the homogeneity property in Remark~\ref{rmk:properties:asym:cap}, since $\alpha V_{x,y}\sim\mathcal{E}(\tfrac{\lambda}{\alpha})$ for any $\alpha>0$. 
\end{myex}

\begin{myex}[Symmetric Dirichlet distribution] \label{ex:uniform:simplex}
Consider a DMC that is described by an $n\times n$ channel transition matrix, whose rows $\Wn_{x,\cdot}$ are independent random variables on the $n$-simplex. More precisely, let the rows $\Wn_{x,\cdot}$ be i.i.d. random variables according to the symmetric Dirichlet distribution $\mathsf{Dir}(\lambda, \hdots, \lambda)$ with concentration parameter $\lambda>0$. It is known \cite[Theorem.~4.1, p.~594]{ref:Devroye-86} that for $n$ exponentially distributed i.i.d. random variables $V_{x,1},\hdots, V_{x,n} \sim\mathcal{E}(\lambda)$, the multivariate random variable $\Wn_{x,\cdot}:=V_{x,\cdot} / \sum_{y \in [n]}V_{x,y}$ admits a symmetric Dirichlet distribution $\mathsf{Dir}(\lambda, \hdots, \lambda)$, that is the uniform distribution over the $n$-simplex for $\lambda=1$. Hence, by Example~\ref{ex:exp} the capacity of a channel $\Wn$ with i.i.d. symmetric Dirichlet distributed rows converges to $\tfrac{1-\kappa}{\ln 2}$ almost surely and in $\Lp{2}$ as $n\to\infty$, where $\kappa$ denotes Euler's constant.
\end{myex}

\begin{myex}[Lognormal distribution] \label{ex:lognormal1}
Consider a DMC (with $n=m$) as defined above using a lognormal distribution $\exp \mathcal{N}(z,\sigma)$ with parameters $z\in\R$ and $\sigma>0$. Then for $n \to \infty$ its capacity converges to $\tfrac{\sigma^{2}}{2 \ln 2}$ almost surely and in $\Lp{2}$. This follows directly from Theorem~\ref{thm:main}, since for $V_{x,y}\sim\exp \mathcal{N}(z,\sigma)$ we have $\mu_1=\E{V_{x,y}}=\exp(z+\tfrac{\sigma^{2}}{2})$ and $\mu_2=\E{V_{x,y}\log V_{x,y}}=\tfrac{z+\sigma^{2}}{\ln 2}\exp(z+\tfrac{\sigma^{2}}{2})$. We note that $\alpha V_{x,y}\sim\exp \mathcal{N}(z + \ln \alpha,\sigma)$ for positive $\alpha$, which by the homogeneity property (cf.~Remark~\ref{rmk:properties:asym:cap}) implies that the asymptotic capacity does not depend on $z$.
\end{myex}
Four additional examples considering the uniform, gamma, chi-squared and beta distribution can be found in Appendix~\ref{app:additional:examples}.
Before we present a rigorous proof of Theorem~\ref{thm:main} in the next section let us sketch an informal motivation, that might provide some intuition about the proof.

Let us assume that the i.i.d.~random variables $V_{x,y}$ take values in a finite set $[k]$, for some $k\in\mathbb{N}$. Statistically as the input and output alphabet get larger (i.e., $ n,m\gg k$), the channel matrix $\Wn$ resembles a weakly symmetric channel (i.e., every row is a permutation of every other row and all the column sums are equal). It is known \cite[Theorem~7.2.1]{cover}, that the capacity of a weakly symmetric channel $\Wn$ is given by $\log m - H(\Wn_{x,.})$ for $x\in[n]$ and that the uniform input distribution is capacity achivieng, i.e., the optimal input distribution does not depend on the channel realization. We further note that the capacity of such channels only depends on the statistics of the channel entries.
In Section~\ref{sec:pfmain}, to prove Theorem~\ref{thm:main}, we derive an analytical upper and lower bound for the capacity and show that in the limit $n\to \infty$ they coincide at the value predicted by Theorem~\ref{thm:main}. The upper bound is shown to be $\log m - \max_{x\in[n]} H(\Wn_{x,.})$ and the lower bound $\I{\hat{p}}{\Wn}$, where $\hat{p}$ is the uniform distribution on $[n]$.

\subsection{Proof of Theorem~\ref{thm:main}} \label{sec:pfmain}
To keep the notation simple we denote the channel transition matrix $\Wn$ by $\W$.
We reformulate the problem~\eqref{eq:capacity} by introducing an additional decision variable $q \in \Delta_m$ representing the output distribution of the channel, together with the coupling constraint $\W\transp p = q$. Whereas the Lagrange dual problem to \eqref{eq:capacity} can only be implicitly expressed through the solution of a system of linear equations (as reported in \cite{chiang04, chiang05}), introducing the new decision variable $q$ allows us to derive an explicit and simple Lagrange dual problem.
It can be shown (see e.g. \cite[Lemma 1]{ref:Sutter-15}) that the optimization problem \eqref{eq:capacity} is equivalent to
\begin{equation} \label{opt:primal}
 	\text{(primal program):} \quad \left\{ \begin{array}{lll}
			&\max\limits_{p,q} 		&- r\transp p + H(q) \\
			&\st					& \W\transp p = q\\
			& 					& p\in \Delta_{n}, \ q\in\Delta_{m},
	\end{array} \right.
\end{equation}
where $r_{x}:=-\sum_{y=1}^{m}\W_{x,y}\log \W_{x,y}$.
The Lagrangian dual program to \eqref{opt:primal} is 
\begin{equation} \label{opt:dual}
\text{(dual program):} \quad  \min \limits_{\lambda} \left \lbrace G(\lambda) + F(\lambda) \,\,\, : \,\,\, \lambda \in \R^{m} \right \rbrace,
\end{equation}
where $G,F:\R^m \to \R$ are given by
\begin{align*} 
G(\lambda)= \left\{ \begin{array}{ll}
			\underset{p}{\max} 		&-r\transp p + \lambda\transp \W\transp p \\
			\text{s.t. } 						& p\in\Delta_{n}
	\end{array} \right.
	\quad \textnormal{and} \quad
	F(\lambda)= \left\{ \begin{array}{ll}
			\underset{q}{\max} 		&H(q)-\lambda\transp q \\
			\text{s.t. } 				& q\in\Delta_{m}.
	\end{array}\right. 
\end{align*}
Note that since the coupling constraint $ \W\transp p = q$ in the primal program \eqref{opt:primal} is affine, the set of optimal solutions to the dual program \eqref{opt:dual} is nonempty \cite[Proposition~5.3.1]{ref:Bertsekas-09} and as such the optimum is attained.
As shown in \cite[Section~2]{ref:Sutter-15}, $G$ and $F$ have analytical solutions given as
\begin{equation}
G(\lambda)= \max \limits_{i \in [n]} \left(  \W \lambda - r\right)_i \quad \textnormal{and} \quad F(\lambda)= \log \left( \sum_{i=1}^m 2^{-\lambda_i}\right). \label{eq:solDual}
\end{equation}

\begin{mylem}
Strong duality holds between \eqref{opt:primal} and \eqref{opt:dual}.
\end{mylem}
\begin{proof}
The proof follows by a standard strong duality result of convex optimization, see \citep[Proposition~5.3.1]{ref:Bertsekas-09}.
\end{proof}
Weak duality of convex programming implies that the dual always is an upper bound to the primal problem, i.e.,\ for every $pÊ\in \Delta_n$ and for every $\lambda \in \R^m$, $C_{\LB}^{(p)}(\W):=\I{p}{\W}\leq G(\lambda)+F(\lambda)=:C_{\UB}^{(\lambda)}(\W)$. By following the proof of Lemma~\ref{lem:measurability}, one can show that the mapping $\mathcal{M}_{n}\ni \W\mapsto C_{\UB}^{(\lambda)}(\W)\in \Rp$ is measurable for any $\lambda\in\R^{m}$ and as such $C_{\UB}^{(\lambda)}(\W)$ is a random variable. To prove Theorem~\ref{thm:main}, we consider the upper bound $C_{\UB}^{(\lambda=0)}(\W):=G(0)+F(0)$ which is the Lagrange dual function evaluated at $\lambda=0$. As a lower bound we consider the mutual information evaluated for a uniform input distribution, i.e.,\ $C_{\LB}^{(p\sim\mathcal{U})}(\W):=\I{p}{\W}$, where $p_i = \tfrac{1}{n}$ for all $i \in [n]$.
Note that by the measurability of the mutual information $C_{\LB}^{(p\sim\mathcal{U})}(\W)$ is a random variable.
We will show that $C_{\LB}^{(p\sim\mathcal{U})}(\W)$ and $C_{\UB}^{(\lambda=0)}(\W)$ converge to the asymptotic capacity predicted by Theorem~\ref{thm:main} in the limit $n \to \infty$ which then proves the assertion.

\begin{mylem} \label{lem:UBdav}
Under Assumptions~\ref{ass:finite:moments} and \ref{ass:input:output:size}, for $n\to \infty$, the random variable $C_{\UB}^{(\lambda=0)}(\W)$ converges to $\tfrac{\mu_2}{\mu_1} - \log \mu_1$ almost surely and in $\Lp{2}$.
\end{mylem}
\begin{proof}
According to \eqref{eq:solDual} we have
\begin{align}
C_{\UB}^{(\lambda=0)}(\W) &= G(0) + F(0) \nonumber \\
 &= \max\limits_{x \in [n]}\left\lbrace - r_x \right \rbrace + \log m \nonumber \\
 &= \max\limits_{x \in [n]}\left\lbrace \sum_{y=1}^m \W_{x,y} \log \W_{x,y} \right \rbrace + \log m.  \label{eq:defdUB}
\end{align}
According to Lemma~\ref{lemma:ImpLimit}, for every $x \in [n]$ as $n\to \infty$,  $\sum_{y=1}^m \W_{x,y} \log \W_{x,y} + \log m $ converges to $\tfrac{\mu_2}{\mu_1} - \log \mu_1$ almost surely and in $\Lp{2}$. This finally proves the assertion.
\end{proof}

\begin{mylem} \label{lem:LBdav}
Under Assumptions~\ref{ass:finite:moments} and \ref{ass:input:output:size}, for $n\to \infty$, the random variable $C_{\LB}^{(p\sim\mathcal{U})}(\W)$ converges to $\tfrac{\mu_2}{\mu_1} - \log \mu_1$ almost surely and in $\Lp{2}$.
\end{mylem}
\begin{proof}
The mutual information for a uniform input distribution, i.e., $p_i = \frac{1}{n}$ for all $i \in [n]$ can be written as
\begin{align}
C_{\LB}^{(p\sim\mathcal{U})}(\W)  &= \frac{1}{n} \sum_{x\in [n],y\in[m]} \W_{x,y} \left(\log n + \log \W_{x,y}  - \log \sum_{k \in [n]} \W_{k,y}  \right) \nonumber \\
&=\frac{1}{n} \sum_{x\in [n],y\in [m]} \W_{x,y} \left(\log n + \log \W_{x,y} \right) - \frac{1}{n} \sum_{x,y \in [n]} \W_{x,y} \log \sum_{k\in[n]} \W_{k,y}.  \label{eq:dav}
\end{align}
According to Lemma~\ref{lemma:normalization}, for $n\to \infty$, $ \frac{1}{n} \sum_{x \in [n],y\in [m]} \W_{x,y} \log \sum_{k\in[n]} \W_{k,y}$ converges to $-\log \gamma$ almost surely and in $\Lp{2}
$. We can simplify the first part of \eqref{eq:dav} by making use of the fact that $\W_{x,y}$ is normalized, i.e., that $\sum_{y\in[m]} \W_{x,y}=1$ for all $x\in [n]$,
\begin{align}
\frac{1}{n} \sum_{x\in [n],y\in [m]} \W_{x,y} \left(\log n + \log \W_{x,y} \right) &= \frac{1}{n} \left(\log n \sum_{x\in[n]} \sum_{y \in [m]} \W_{x,y}+ \sum_{x \in [n],y\in [m]} \W_{x,y} \log \W_{x,y} \right) \nonumber \\
&= \log n + \frac{1}{n} \sum_{x \in [n],y\in [m]} \W_{x,y} \log \W_{x,y}. \label{eq:dav2}
\end{align}
Consider the upper bound 
\begin{align}
\log n + \frac{1}{n} \sum_{x \in [n],y\in [m]} \W_{x,y} \log \W_{x,y} &\leq \log n + \max \limits_{x\in [n]} \sum_{y \in [m]} \W_{x,y} \log \W_{x,y} \nonumber \\
 &= \log n + \sum_{y \in [m]} \W_{\tilde x,y} \log \W_{\tilde x,y} \nonumber \\
 &= \log m +  \sum_{y \in [m]} \W_{\tilde x,y} \log \W_{\tilde x,y} - \log\left( \gamma + \frac{\varepsilon_n}{n} \right) , \label{eq:rhs1}
\end{align}
for some $\tilde x \in [n]$, where $\varepsilon_{n}:= \ceil{\gamma n}-\gamma n \in[0,1)$ for all $n$. According to Lemma~\ref{lemma:ImpLimit}, the right hand side of \eqref{eq:rhs1} converges to $\tfrac{\mu_2}{\mu_1} - \log \mu_1-\log \gamma$ almost surely and in $\Lp{2}$ for $n\to \infty$. We can also bound the same term from below as 
\begin{align}
\log n + \frac{1}{n} \sum_{x\in [n], y \in [m]} \W_{x,y} \log \W_{x,y} &\geq\log n + \min \limits_{x\in [n]} \sum_{y \in [m]} \W_{x,y} \log \W_{x,y} \nonumber  \\
 &= \log n + \sum_{y \in [m]} \W_{\bar x,y} \log \W_{\bar x,y} \nonumber \\
 &=\log m + \sum_{y \in [m]} \W_{\bar x,y} \log \W_{\bar x,y} - \log\left( \gamma + \frac{\varepsilon_n}{n} \right), \label{eq:rhs2}
\end{align}
for some $\bar x \in [n]$, where $\varepsilon_{n}:=\ceil{\gamma n}-\gamma n \in[0,1)$ for all $n$. According to Lemma~\ref{lemma:ImpLimit}, the right hand side of \eqref{eq:rhs2} converges to $\tfrac{\mu_2}{\mu_1} - \log \mu_1-\log \gamma$ almost surely and in $\Lp{2}$ as $n\to \infty$.
Thus for $n\to \infty$, \eqref{eq:dav} converges to $\tfrac{\mu_2}{\mu_1} - \log \mu_1$ in $\Lp{2}$ which proves the assertion.
\end{proof}
Lemmas~\ref{lem:UBdav} and \ref{lem:LBdav} complete the proof of Theorem~\ref{thm:main} as $C_{\LB}^{(p\sim\mathcal{U})}(\W) \leq C(\W) \leq C_{\UB}^{(\lambda =0)}(\W)$.


\section{Simulation results} \label{sec:simulation}
In this section we compute the capacity of the DMCs introduced in Section~\ref{sec:main} for finite alphabet sizes. For the computation we use a recently introduced method \cite{ref:Sutter-15} which allows us to efficiently compute close upper and lower bounds to the capacity. Roughly speaking, the method \cite{ref:Sutter-15} is an iterative accelerated first-order method that exploits duality of convex programming together with the fact that entropy maximization problems admit closed-form solutions.

\begin{myex}[Exponential distribution] \label{ex:simu:exp}
We consider a channel that is given by the stochastic matrix  $\W=(\W_{x,y})_{x,y\in[n]}$ with $\W_{x,y}=V_{x,y} / \sum_{y \in [n]}V_{x,y}$, where $V_{x,y}$ are i.i.d. $\mathcal{E}(\lambda)$ random variables with $\lambda=\tfrac{1}{10}$ for all $x,y\in[n]$. 
As explained in Example~\ref{ex:uniform:simplex} with this channel construction the rows $\W_{x,\cdot}$ admit a symmetric Dirichlet distribution with concentration parameter $\lambda=\tfrac{1}{10}$ for all $x\in[n]$.
Figure~\ref{fig:plotExp} depicts the capacity of $\W$ for variable alphabet sizes. We perform five independent experiments for each value of $n$. On can observe that as $n \to \infty$ the capacity approaches the asymptotic limit as determined in Example~\ref{ex:exp}.  In addition one can see that the variance between the capacity of the two independently chosen channels is decreasing for increasing alphabet sizes.

\begin{figure}[!htb]
  \begin{tikzpicture}
	\begin{axis}[
		height=9cm,
		width=16cm,
		grid=major,
		xlabel=alphabet size $n$,
		ylabel=capacity \bracket{bits per channel use},
		xmin=1,
		xmax=1000,
		ymax=0.723,
		ymin=0.588,
		legend style={at={(0.79,0.965)},anchor=north,legend cell align=left} 
		]


	\addlegendimage{only marks,black,mark=square*,mark size=1pt} 	
	\addlegendimage{only marks,black,mark=*,mark size=1pt} 
	\addlegendimage{black,thick,smooth,dashed,mark size=1pt}

		\addplot[blue,mark=square*,only marks,mark options={fill=blue},mark size=1pt] coordinates {
		(10,0.6517)	
		(50,0.6911)
		(100,0.6636)
		(150,0.6667)
		(200,0.6804)	
		(250,0.6714)	
		(300,0.6682)	
		(350,0.6622)	
		(400,0.6603)	
		(450,0.6567)	
		(500,0.6646)	
		(550,0.6625)
		(600,0.6568)
		(650,0.6563)
		(700,0.6550)
		(750,0.6537)
		(800,0.6520)
		(850,0.6562)
		(900,0.6546)
		(950,0.6515)
		(1000,0.6480) 
	};

		\addplot[blue,mark=*,only marks,mark options={fill=blue},mark size=1pt] coordinates {
		(10,0.6514)	
		(50,0.6896)
		(100,0.6630)
		(150,0.6645)
		(200,0.6795)	
		(250,0.6699)	
		(300,0.6675)	
		(350,0.6593)	
		(400,0.6595)	
		(450,0.6550)	
		(500,0.6640)	
		(550,0.6611)
		(600,0.6556)
		(650,0.6541)
		(700,0.6540)
		(750,0.6521)
		(800,0.6512)
		(850,0.6549)
		(900,0.6540)
		(950,0.6505)
		(1000,0.6480) 
	};
		
	
			\addplot[red,mark=square*,only marks,mark options={fill=red},mark size=1pt] coordinates {
		(10,0.6187)	
		(50,0.6929)
		(100,0.6984)
		(150,0.6791)
		(200,0.6703)	
		(250,0.6799)	
		(300,0.6635)	
		(350,0.6656)	
		(400,0.6629)	
		(450,0.6682)	
		(500,0.6666)	
		(550,0.6643)
		(600,0.6583)
		(650,0.6552)
		(700,0.6551)
		(750,0.6542)
		(800,0.6519)
		(850,0.6562)
		(900,0.6540)
		(950,0.6515)
		(1000,0.6480) 
	};

		\addplot[red,mark=*,only marks,mark options={fill=red},mark size=1pt] coordinates {
		(10,0.6172)	
		(50,0.6916)
		(100,0.6979)
		(150,0.6775)
		(200,0.6695)	
		(250,0.6786)	
		(300,0.6625)	
		(350,0.6643)	
		(400,0.6623)	
		(450,0.6670)	
		(500,0.6658)	
		(550,0.6628)
		(600,0.6574)
		(650,0.6537)
		(700,0.6544)
		(750,0.6525)
		(800,0.6512)
		(850,0.6549)
		(900,0.6534)
		(950,0.6505)
		(1000,0.6480) 
	};
	
				\addplot[ForestGreen,mark=square*,only marks,mark options={fill=ForestGreen},mark size=1pt] coordinates {
		(10,0.6972)	
		(50,0.7089)
		(100,0.6864)
		(150,0.6786)
		(200,0.6892)	
		(250,0.6753)	
		(300,0.6690)	
		(350,0.6814)	
		(400,0.6696)	
		(450,0.6685)	
		(500,0.6672)	
		(550,0.6585)
		(600,0.6588)
		(650,0.6555)
		(700,0.6556)
		(750,0.6573)
		(800,0.6524)
		(850,0.6563)
		(900,0.6534)
		(950,0.6502)
		(1000,0.6485)																																			
	};
	
				\addplot[ForestGreen,mark=*,only marks,mark options={fill=ForestGreen},mark size=1pt] coordinates {
		(10,0.6967)	
		(50,0.7039)
		(100,0.6833)
		(150,0.6756)
		(200,0.6869)	
		(250,0.6739)	
		(300,0.6660)	
		(350,0.6791)	
		(400,0.6681)	
		(450,0.6674)	
		(500,0.6658)	
		(550,0.6573)
		(600,0.6576)
		(650,0.6542)
		(700,0.6543)
		(750,0.6561)
		(800,0.6510)
		(850,0.6549)
		(900,0.6520)
		(950,0.6493)
		(1000,0.6478)																																			
	};

				\addplot[YellowOrange,mark=square*,only marks,mark options={fill=YellowOrange},mark size=1pt] coordinates {
		(10,0.7224)	
		(50,0.6976)
		(100,0.6953)
		(150,0.6812)
		(200,0.6831)	
		(250,0.6848)	
		(300,0.6713)	
		(350,0.6719)	
		(400,0.6639)	
		(450,0.6618)	
		(500,0.6538)	
		(550,0.6592)
		(600,0.6584)
		(650,0.6595)
		(700,0.6552)
		(750,0.6564)
		(800,0.6549)
		(850,0.6529)
		(900,0.6529)
		(950,0.6496)
		(1000,0.6485)																																			
	};

				\addplot[YellowOrange,mark=*,only marks,mark options={fill=YellowOrange},mark size=1pt] coordinates {
		(10,0.7221)	
		(50,0.6919)
		(100,0.6905)
		(150,0.6788)
		(200,0.6802)	
		(250,0.6821)	
		(300,0.6695)	
		(350,0.6704)	
		(400,0.6622)	
		(450,0.6602)	
		(500,0.6530)	
		(550,0.6580)
		(600,0.6576)
		(650,0.6587)
		(700,0.6543)
		(750,0.6554)
		(800,0.6536)
		(850,0.6517)
		(900,0.6520)
		(950,0.6487)
		(1000,0.6478)																																			
	};

				\addplot[SkyBlue,mark=square*,only marks,mark options={fill=SkyBlue},mark size=1pt] coordinates {
		(10,0.5904)	
		(50,0.7190)
		(100,0.6727)
		(150,0.6765)
		(200,0.6859)	
		(250,0.6649)	
		(300,0.6739)	
		(350,0.6688)	
		(400,0.6679)	
		(450,0.6634)	
		(500,0.6613)	
		(550,0.6581)
		(600,0.6610)
		(650,0.6572)
		(700,0.6592)
		(750,0.6571)
		(800,0.6549)
		(850,0.6565)
		(900,0.6545)
		(950,0.6524)
		(1000,0)																																			
	};	

				\addplot[SkyBlue,mark=*,only marks,mark options={fill=SkyBlue},mark size=1pt] coordinates {
		(10,0.5887)	
		(50,0.7129)
		(100,0.6682)
		(150,0.6735)
		(200,0.6830)	
		(250,0.6625)	
		(300,0.6721)	
		(350,0.6668)	
		(400,0.6661)	
		(450,0.6626)	
		(500,0.6603)	
		(550,0.6572)
		(600,0.6600)
		(650,0.6564)
		(700,0.6581)
		(750,0.6559)
		(800,0.6540)
		(850,0.6552)
		(900,0.6537)
		(950,0.6516)
		(1000,0)																																			
	};	
	

	\addplot[black,thick,smooth,dashed] coordinates {
		(1,0.6099)	
		(1000,0.6099)
	};
		

\addlegendentry{ \, upper bound $i$th experiment}
\addlegendentry{ \, lower bound $i$th experiment}
\addlegendentry{ \, asymptotic capacity}

	\end{axis}
\end{tikzpicture}

\caption{For different alphabet sizes $n$ we plot the capacity of five random channels, constructed as explained in Example~\ref{ex:simu:exp}. The method introduced in \cite{ref:Sutter-15} is used to determine upper and lower bounds for the capacity for finite alphabet sizes $n$. The asymptotic capacity (for $n\to \infty$) is depticted by the dashed line.}
\label{fig:plotExp}
\end{figure}
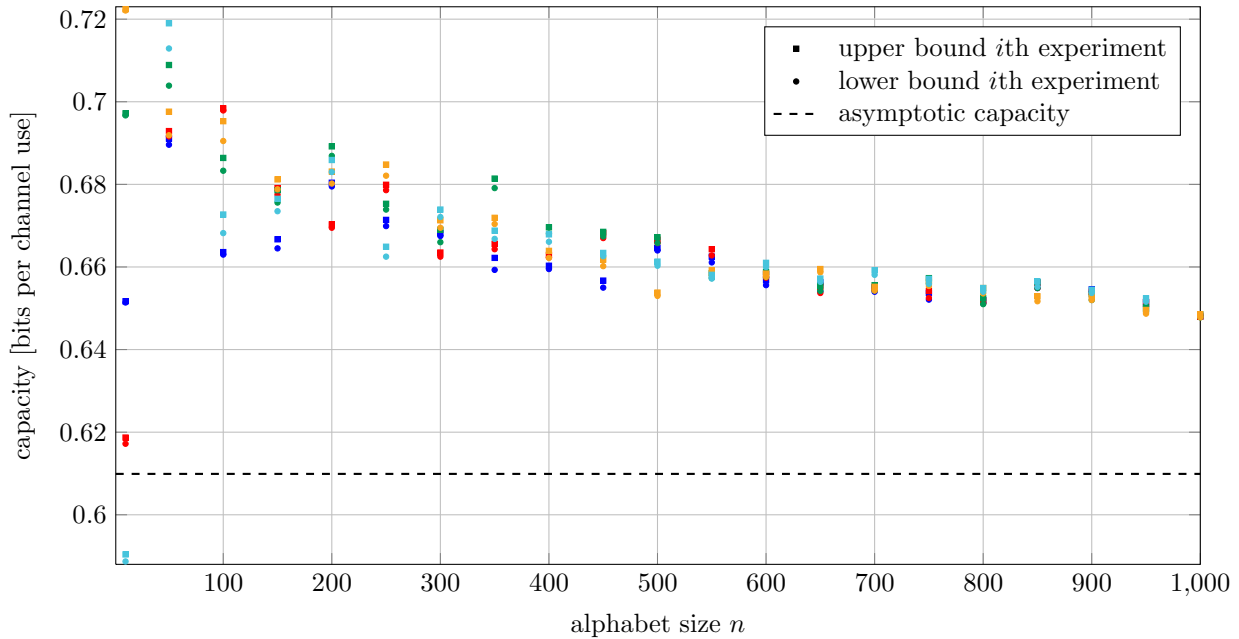

\end{myex}

\section{Application in Bayesian Experiment Design} \label{sec:application}

The main objective of optimal experiment design is, based on prior knowledge, to select a most informative experiment, where 
we restrict attention to a certain notion of information that traces back to Shannon~\cite{shannon48}; see \cite{ref:Lindley-56} for a comprehensive survey. This section will motivate the study of a convergence rate for the asymptotic capacity (under more restrictive assumptions on the channel model) that is the content of Section~\ref{sec:finiteSize}.

Let the random variable $X\in\mathcal{X}:=[n]$ describe a parameter to be determined with a prior probability distribution $p\in\Delta_{n}$ and let the random variable $Y\in\mathcal{Y}:=[m]$ denote an observation.
Furthermore, consider a family of experiments $(\W^{(\lambda,n)})_{\lambda \in \Lambda}$, where $\Lambda\subset \R^{d}$ characterizes the set of all admissible experiments, and each experiment $\W^{(\lambda,n)} \in \mathcal{M}_{n,m}$ is characterized by the conditional probabilities $\W^{(\lambda,n)}_{x,y}:=\mathds{P}^{(\lambda)}[Y=y|X=x]$ for all $x\in\mathcal{X}$ and $y\in\mathcal{Y}$.\footnote{Strictly speaking, an experiment consists of the tuple $\{ \mathcal{Y}, \Borelsigalg{\mathcal{Y}},\mathcal{X}, \W^{(\lambda,n)}, \mathds{P}^{(\lambda)} \}$ as pointed out in \cite{ref:Lindley-56}. Since the conditional probability $\W^{(\lambda,n)}$ is our optimization variable and since $\mathcal{X}$ and $\mathcal{Y}$ remain constant, with a slight abuse of notation, we call $\W^{(\lambda,n)}$ an experiment.} The task of optimal experiment design is, given a prior distribution $p\in\Delta_{n}$, to find the experiment that provides the highest average amount of information, as described by the mutual information between the parameter and the observation \cite[Definition~2]{ref:Lindley-56}, i.e., the goal is to find $\lambda^{\star}\in\Lambda$ such that $\I{p}{\W^{(\lambda^{\star},n)}} \geq \I{p}{\W^{(\lambda,n)}}$ for all $\lambda\in\Lambda$. This requires one to compute
\begin{equation} \label{eq:experimental:design}
\sup_{\lambda\in\Lambda} \I{p}{\W^{(\lambda,n)}}.
\end{equation}
The optimization problem \eqref{eq:experimental:design} in general is difficult to solve. Moreover, an evaluation of the objective function, the mutual information, for a given $\lambda$ has a computational complexity of $O(nm)$ and as such for large sets $\mathcal{X}$ and $\mathcal{Y}$ even solving \eqref{eq:experimental:design} for local optimality can be computationally demanding.

The task of designing optimal experiments has recently attracted interest in the context of biological systems, where understanding about the underlying biological mechanisms emerges through iterations of modelling and experiments. Since experiments are expensive an effective selection of informative experiments is essential, see \cite{ref:Busetto-13}. We will show in Section~\ref{sec:upper:bound:experiment} that the asymptotic capacity formula, given in Theorem~\ref{thm:main}, allows us to derive upper bounds on the expected information gain by an experiment for certain classes of (random) experiments. In addition, Theorem~\ref{thm:main} provides an efficient method to select suboptimal experiments, that are almost optimal in our numerical example, see Example~\ref{ex:expDesign}.
Let $(V_{x,y}^{(\lambda)})_{x\in[n], y\in[m]}$ be i.i.d.\ random variables for each $\lambda\in\Lambda\subset\R^d$ and consider a channel transition matrix $\W_{x,y}^{(\lambda,V,n)}=V_{x,y}^{(\lambda)} / \sum_{y \in [m]}V_{x,y}^{(\lambda)}$.

%

\subsection{Upper bound on maximum expected information gain}\label{sec:upper:bound:experiment}
In the limit, as $n\to\infty$, we can establish the following upper bound on the maximum expected information gain by an experiment.
\begin{myprop}[Upper bound on maximum expected information gain] \label{prop:upper:bound}
For the family of channels $( \W^{(\lambda,V,n)})_{\lambda \in \Lambda}$ introduced above that satisfy Assumptions~\ref{ass:input:output:size} and \ref{ass:convergence:rate}, we have with high probability
\begin{align} \label{eq:upper:bound:MI}
\lim_{n\to\infty} \sup_{\lambda\in\Lambda}  \I{p}{ \W^{(\lambda,V,n)}} &\leq \sup_{\lambda\in\Lambda} \frac{\Ent(V_{x,y}^{(\lambda)})}{\E{V_{x,y}^{(\lambda)}}}.
\end{align}
\end{myprop}
The upper bound provided by Proposition~\ref{prop:upper:bound} is particularly useful if the right-hand side of \eqref{eq:upper:bound:MI} admits a closed form solution, whereas the optimal information gain $\sup_{\lambda\in\Lambda}  \I{p}{ \W^{(\lambda,V,n)}}$ is difficult to compute (see Example~\ref{ex:expDesign} for more details). \\

Before proving Proposition~\ref{prop:upper:bound} we recall a preliminary standard result.
\begin{mylem}[Theorem~7.11 in \cite{ref:Rudin-76}] \label{lem:interchange:limits}
Suppose $X$ is a metric space, $E$ is a subset of $X$ and $x$  is a limit point of $E$. Suppose $f_{n}:X\to\R$ for each $n\in\mathbb{N}$ and $f:X\to\R$ are functions and $A_{n}$ are numbers. If $\lim_{n\to\infty} f_{n}(x) = f(x)$ uniformly in $X$ and $\lim_{y\to x}f_{n}(y) = A_{n}$ pointwise over $n \in \mathbb{N}$. Then
\begin{equation*}
\lim_{y\to x}\lim_{n\to \infty} f_{n}(y) = \lim_{n\to\infty} \lim_{y\to x}f_{n}(y).
\end{equation*}
\end{mylem}
\begin{proof}[Proof of Proposition~\ref{prop:upper:bound}]
We show that with high probability
\begin{align} \label{eq:upper:bound:MI:proof}
\lim_{n\to\infty} \sup_{\lambda\in\Lambda}  \I{p}{ \W^{(\lambda,V,n)}} &\leq \lim_{n\to\infty} \sup_{\lambda\in\Lambda} C( \W^{(\lambda,V,n)} ) = \sup_{\lambda\in\Lambda} \lim_{n\to\infty}   C( \W^{(\lambda,V,n)} ) = \sup_{\lambda\in\Lambda} \frac{\Ent(V_{x,y}^{(\lambda)})}{\E{V_{x,y}^{(\lambda)}}}.
\end{align}
The first inequality of \eqref{eq:upper:bound:MI:proof} is trivial and the last equality follows by Theorem~\ref{thm:main}. Therefore it remains to prove that the first equality in \eqref{eq:upper:bound:MI:proof} holds almost surely. 
Note first that the following property holds
\begin{enumerate}[(i)] 
\item \label{ass:uniform:convergence} The capacity of the channel $\W^{(\lambda,V,n)}$ converges uniformly in $\Lambda$ to its asymptotic capacity in probability, i.e.,
\begin{equation*}
\text{for all }\varepsilon>0, \quad \lim_{n\to\infty} \Prob{ \sup\limits_{\lambda\in\Lambda}  \left|C( \W^{(\lambda,V,n)} ) - \frac{\Ent(V^{(\lambda)}_{x,y})}{\E{V^{(\lambda)}_{x,y}}} \right|\geq\varepsilon} = 0,
\end{equation*}
\end{enumerate}
because by Theorem~\ref{thm:finiteSize} (whose derivation is provided in the next section), we know that for each $n$ there exists $M_{n}<\infty$ and  $N_{n}\leq 1$ as well as $\Omega_{n}\subset \Omega$ with $\Prob{\Omega_{n}}\geq N_{n}$ such that
\begin{equation}
\sup_{\lambda\in\Lambda}  \left|C( \W^{(\lambda,V,n)} ) - \frac{\Ent(V^{(\lambda)}_{x,y})}{\E{V^{(\lambda)}_{x,y}}}\right| \leq M_{n} \quad \textnormal{on }\Omega_{n}.
\end{equation}
Moreover, $M_{n}\to 0$ and $N_{n}\to 1$ as $n\to\infty$, which implies that 
\begin{equation*}
\text{for all }\varepsilon>0, \quad \lim_{n\to\infty} \Prob{ \sup\limits_{\lambda\in\Lambda}  \left|C( \W^{(\lambda,V,n)} ) - \frac{\Ent(V^{(\lambda)}_{x,y})}{\E{V^{(\lambda)}_{x,y}}} \right|\geq\varepsilon} = 0
\end{equation*}
and hence, property~\eqref{ass:uniform:convergence} holds. Note also that the following property holds trivially since $C( \W ) \leq \log( n \wedge m)$ for any channel matrix $\W\in\mathcal{M}_{n,m}$
\begin{enumerate}[(ii)] 
\item \label{ass:ii:family1}  $\sup\limits_{\lambda\in\Lambda} C( \W^{(\lambda,V,n)} ) < \infty$ almost surely for all $n\in\mathbb{N}$.
\end{enumerate}
Hence Lemma~\ref{lem:interchange:limits}, using the two properties \eqref{ass:uniform:convergence}, (\ref{ass:ii:family1}\ref{ass:ii:family1}), implies that with high probability
\begin{equation*}
 \lim_{n\to\infty} \sup_{\lambda\in\Lambda} C( \W^{(\lambda,V,n)} ) = \sup_{\lambda\in\Lambda} \lim_{n\to\infty}   C( \W^{(\lambda,V,n)} ),
\end{equation*}
which readily can be shown to imply the desired equality and therefore completes the proof.
\end{proof}


\subsection{Example: Constrained lognormal distribution} \label{ex:expDesign}
We consider the setting given in Example~\ref{ex:lognormal1} and introduce a parameter $\lambda:=(z,\sigma^{2})\in\R\times\Rsp$. For given constants $\ell_{i}, u_{i}$ for $i=1,2$ and $\ell_{1}>0$, we consider the family of experiments $(\W^{(\lambda,V,n)})_{\lambda\in\Lambda}$, where
\begin{equation}\label{eq:familiy:channels:ex:lognormal}
\Lambda = \left\{(z,\sigma^2) \in \R \times \Rp : \exp(z+\tfrac{\sigma^{2}}{2}) \in [\ell_1, u_1],  (\exp(\sigma^{2})-1)\exp(2z + \sigma^{2}) \in [\ell_2,u_2]\right\},
\end{equation}
and $\E{V_{x,y}^{(\lambda)}} = \exp(z+\tfrac{\sigma^{2}}{2})$ and $ \Var{V_{x,y}^{(\lambda)}} = (\exp(\sigma^{2})-1)\exp(2z + \sigma^{2})$. For this family of experiments Assumptions~\ref{ass:input:output:size}, \ref{ass:convergence:rate} clearly hold and an upper bound to the maximum expected information gain provided by an experiment, using Proposition~\ref{prop:upper:bound}, can be stated in closed form.
\begin{myprop} \label{prop:exp:design}
In the limit $n\to\infty$ an upper bound on the maximum information gain by an experiment from the family \eqref{eq:familiy:channels:ex:lognormal} is given with high probability by
\begin{equation*}
\lim_{n\to\infty} \sup_{\lambda\in\Lambda}  \I{p}{ \W^{(\lambda,V,n)}} \leq  (2\ln 2)^{-1} \ln\left( \frac{u_{2}}{\ell^{2}_{1}}+1\right).
\end{equation*}
\end{myprop}
\begin{proof}
According to Proposition~\ref{prop:upper:bound} and Example~\ref{ex:lognormal1} 
\begin{align}
\lim_{n\to\infty} \sup_{\lambda\in\Lambda}  \I{p}{ \W^{(\lambda,V,n)}} &\leq \lim_{n\to\infty} \sup_{\lambda\in\Lambda} C( \W^{(\lambda,V,n)} ) = \sup_{\lambda\in\Lambda} \lim_{n\to\infty}   C( \W^{(\lambda,V,n)} ) \nonumber \\
&=\left\{ \begin{array}{rll}
\max\limits_{\sigma^{2},z}&&\frac{\sigma^{2}}{2\ln 2} \\
\text{s.t.} 	&& \ell_{1}\leq \exp\left( z+\frac{\sigma^{2}}{2} \right) \leq u_{1} \\
			&& \ell_{2}\leq \left( \exp(\sigma^{2})-1 \right) \exp\left( 2z + \sigma^{2}\right) \leq u_{2}\\
			&& \sigma^{2}\in\Rp, \ z\in\R.
\end{array} \right.  \label{eq:proof:prop:exp:design}
\end{align}
By introducing the variable $\alpha:=\exp(2z+\sigma^{2})$, the optimization problem \eqref{eq:proof:prop:exp:design} can be rewritten as
\begin{align*}
\left\{ \begin{array}{rll}
\max\limits_{\sigma^{2},\alpha}&&\frac{\sigma^{2}}{2\ln 2} \\
\text{s.t.} 	&& \ell_{1}\leq \sqrt{\alpha} \leq u_{1} \\
			&& \ell_{2}\leq \left( \exp(\sigma^{2})-1 \right)\alpha \leq u_{2}\\
			&& \sigma^{2}\in\Rp, \ \alpha\in\Rp.
\end{array} \right. 
\end{align*}
The monotonicity of $\left( \exp(\sigma^{2})-1 \right)\alpha$ with respect to $\sigma^{2}$ implies that the optimizers are uniquely given by
$\alpha=\ell_{1}^{2}$ and $\sigma^{2}=\ln\left( \tfrac{u_{2}}{\ell_{1}^{2}}+1\right)$, which completes the proof.
\end{proof}

Let us consider the case where the prior distribution is uniform. In this case the upper bound \eqref{eq:upper:bound:MI} is tight, by following the proofs of Theorem~\ref{thm:main} and Proposition~\ref{prop:exp:design}. 
Figure~\ref{fig:plotExperimentMeanVar} depicts for different alphabet sizes $n$ in (a) the empirical mean of the maximum expected information gain (blue line) for $1000$ experiments, which in general is difficult to compute in particular for higher dimensional examples than Example~\ref{ex:expDesign}. The red line represents the empirical mean of the suboptimal expected information gain, that is given by evaluating the mutual information for the optimal parameters for the asymptotic capacity, derived in Proposition~\ref{prop:exp:design} and as such is computationally much cheaper.
The empirical variance of the maximum expected information gain (blue line) as well as the empirical variance of the suboptimal expected information gain (red line) are depicted in (b). 
\begin{figure}[!htb]
\hspace{-5mm}
    \subfloat[empirical mean]{\begin{tikzpicture}

\begin{axis}[%
		height=5.5cm,
		width=8cm,
ytick={0.76,0.77,0.78,0.79,0.8},
grid=major,
style={font=\scriptsize},
xmin=0,
xmax=3000,
xlabel={alphabet size $n$},
ymin=0.76,
ymax=0.80,
ylabel={capacity \bracket{bits per channel use}},
legend style={at={(0.63,0.37)},anchor=north,legend cell align=left, font=\scriptsize} 
]

\addplot [color=blue,solid,line width=0.8pt]
  table[row sep=crcr]{%
100	0.762108399513965\\
150	0.771814632463903\\
200	0.777646581927556\\
250	0.779684053112364\\
300	0.782020177974366\\
350	0.78354927136745\\
400	0.784448256568421\\
450	0.785793819798844\\
500	0.786182497237696\\
550	0.786494063900476\\
600	0.787103409191192\\
650	0.787216452634007\\
700	0.78787551534967\\
750	0.788218141962932\\
800	0.78841811681213\\
850	0.788637655138135\\
900	0.788841272893948\\
950	0.788967926353195\\
1000	0.789233474744941\\
1050	0.789352745559591\\
1100	0.789550561572549\\
1150	0.789605448544222\\
1200	0.789756147675632\\
1250	0.789842893844141\\
1300	0.789827789682466\\
1350	0.79004737774738\\
1400	0.790208894287152\\
1450	0.790245581682686\\
1500	0.790371079040314\\
1550	0.790368065909647\\
1600	0.79045210619709\\
1650	0.790496358617014\\
1700	0.790612036036704\\
1750	0.790510564099034\\
1800	0.790677619107404\\
1850	0.790752716154192\\
1900	0.790770163268169\\
1950	0.790820205121219\\
2000	0.790845468972332\\
2050	0.790876543727239\\
2100	0.790985538972348\\
2150	0.790950562652985\\
2200	0.791033346074424\\
2250	0.790991035983941\\
2300	0.791028247383146\\
2350	0.791062642059069\\
2400	0.791070518979082\\
2450	0.791136153496023\\
2500	0.791179363261587\\
2550	0.791151332616846\\
2600	0.791248997367975\\
2650	0.791221312979441\\
2700	0.791259666086976\\
2750	0.791258415548123\\
2800	0.791332137281057\\
2850	0.791316955652315\\
2900	0.791401111116308\\
2950	0.791379939029826\\
3000	0.791373407873675\\
};
\addlegendentry{optimal empirical mean};


\addplot [color=red,solid,line width=0.8pt]
  table[row sep=crcr]{%
100	0.763278758441075\\
150	0.772889147435455\\
200	0.77665424631823\\
250	0.779927834511907\\
300	0.781968205216288\\
350	0.783131708088643\\
400	0.784414241488204\\
450	0.785458293769529\\
500	0.786171224493547\\
550	0.786441354000282\\
600	0.787026684941193\\
650	0.787524968583831\\
700	0.787818892480694\\
750	0.788182916876982\\
800	0.788806739667847\\
850	0.78864239797069\\
900	0.788918491999829\\
950	0.789003381089714\\
1000	0.789295232390277\\
1050	0.789428525391136\\
1100	0.789589346728993\\
1150	0.789694961113276\\
1200	0.789712838518156\\
1250	0.78998764038838\\
1300	0.790123816032406\\
1350	0.790002938748658\\
1400	0.790150556925986\\
1450	0.790225708935697\\
1500	0.79022716138396\\
1550	0.790339313153371\\
1600	0.790421487443388\\
1650	0.790565876270982\\
1700	0.790566681844511\\
1750	0.790593580499492\\
1800	0.790531919138263\\
1850	0.790656262896865\\
1900	0.790782052963652\\
1950	0.790827307677327\\
2000	0.790808042115689\\
2050	0.790912039401512\\
2100	0.790888598498345\\
2150	0.790936623321527\\
2200	0.79100446964323\\
2250	0.790946285303355\\
2300	0.790990313191123\\
2350	0.790982399706956\\
2400	0.79109878790104\\
2450	0.791188548892717\\
2500	0.791148566800311\\
2550	0.791209295166598\\
2600	0.791164639953928\\
2650	0.791260503930899\\
2700	0.791253519971377\\
2750	0.791287847613786\\
2800	0.791211058305156\\
2850	0.791376080952753\\
2900	0.791301850162839\\
2950	0.791390623937265\\
3000	0.791374141665185\\
};

\addlegendentry{suboptimal empirical mean};

\addplot [color=black,dashed,line width=1.0pt]
  table[row sep=crcr]{%
    0	0.792481250360578\\  
  50	0.792481250360578\\  
100	0.792481250360578\\
150	0.792481250360578\\
200	0.792481250360578\\
250	0.792481250360578\\
300	0.792481250360578\\
350	0.792481250360578\\
400	0.792481250360578\\
450	0.792481250360578\\
500	0.792481250360578\\
550	0.792481250360578\\
600	0.792481250360578\\
650	0.792481250360578\\
700	0.792481250360578\\
750	0.792481250360578\\
800	0.792481250360578\\
850	0.792481250360578\\
900	0.792481250360578\\
950	0.792481250360578\\
1000	0.792481250360578\\
1050	0.792481250360578\\
1100	0.792481250360578\\
1150	0.792481250360578\\
1200	0.792481250360578\\
1250	0.792481250360578\\
1300	0.792481250360578\\
1350	0.792481250360578\\
1400	0.792481250360578\\
1450	0.792481250360578\\
1500	0.792481250360578\\
1550	0.792481250360578\\
1600	0.792481250360578\\
1650	0.792481250360578\\
1700	0.792481250360578\\
1750	0.792481250360578\\
1800	0.792481250360578\\
1850	0.792481250360578\\
1900	0.792481250360578\\
1950	0.792481250360578\\
2000	0.792481250360578\\
2050	0.792481250360578\\
2100	0.792481250360578\\
2150	0.792481250360578\\
2200	0.792481250360578\\
2250	0.792481250360578\\
2300	0.792481250360578\\
2350	0.792481250360578\\
2400	0.792481250360578\\
2450	0.792481250360578\\
2500	0.792481250360578\\
2550	0.792481250360578\\
2600	0.792481250360578\\
2650	0.792481250360578\\
2700	0.792481250360578\\
2750	0.792481250360578\\
2800	0.792481250360578\\
2850	0.792481250360578\\
2900	0.792481250360578\\
2950	0.792481250360578\\
3000	0.792481250360578\\
};
\addlegendentry{asymptotic capacity};

\end{axis}
\end{tikzpicture} }
    \subfloat[empirical variance]{\begin{tikzpicture}

\begin{axis}[%
		height=5.5cm,
		width=8cm,
ytick={0,0.0001},
grid=major,
style={font=\scriptsize},
xmin=0,
xmax=3000,
xlabel={alphabet size $n$},
ymin=0,
ymax=0.0001,
ylabel={capacity \bracket{bits per channel use}},
legend style={at={(0.605,0.98)},anchor=north,legend cell align=left, font=\scriptsize} 
]

\addplot [color=blue,solid,line width=0.8pt]
  table[row sep=crcr]{%
100	0.000317793216139653\\
150	0.000168175567452929\\
200	0.000101539495285729\\
250	5.99960875851999e-05\\
300	4.58293799840098e-05\\
350	3.65784518208415e-05\\
400	2.69549516836194e-05\\
450	2.28126115491737e-05\\
500	1.82770324838348e-05\\
550	1.56593150738714e-05\\
600	1.27092106351447e-05\\
650	1.016569001095e-05\\
700	8.86802745346026e-06\\
750	8.02136680300144e-06\\
800	7.28792016103299e-06\\
850	6.23316672510669e-06\\
900	5.57150772736437e-06\\
950	5.1629659186051e-06\\
1000	4.7604404875792e-06\\
1050	4.04716837868226e-06\\
1100	3.616625543631e-06\\
1150	3.48369989608125e-06\\
1200	3.18788516570263e-06\\
1250	3.2318704631601e-06\\
1300	2.57840497999944e-06\\
1350	2.27072893279362e-06\\
1400	2.3832675083039e-06\\
1450	1.94986453757354e-06\\
1500	1.97561214013453e-06\\
1550	2.09275092334472e-06\\
1600	1.67447951556089e-06\\
1650	1.7182168114801e-06\\
1700	1.59007439940689e-06\\
1750	1.59693485152821e-06\\
1800	1.37041615043426e-06\\
1850	1.394093197714e-06\\
1900	1.31499857998536e-06\\
1950	1.2749275194468e-06\\
2000	1.21614390104988e-06\\
2050	1.18969375289165e-06\\
2100	1.04524854112778e-06\\
2150	1.10110968575575e-06\\
2200	9.81767190638464e-07\\
2250	1.02675989075304e-06\\
2300	9.90057668995938e-07\\
2350	8.28397544331108e-07\\
2400	8.54958860194327e-07\\
2450	8.49054677186301e-07\\
2500	7.86907103523839e-07\\
2550	6.40497322609225e-07\\
2600	7.10615118655014e-07\\
2650	6.67312039734126e-07\\
2700	6.97983653658938e-07\\
2750	6.94512656366777e-07\\
2800	7.21470649707585e-07\\
2850	5.66640874452845e-07\\
2900	5.4992797479453e-07\\
2950	5.53806915503786e-07\\
3000	5.76857881250798e-07\\
};
\addlegendentry{optimal empirical variance};


\addplot [color=red,solid,line width=0.8pt]
  table[row sep=crcr]{%
100	0.000321725290238406\\
150	0.000155112512513657\\
200	9.17608864306167e-05\\
250	5.86331084824972e-05\\
300	4.44357881802405e-05\\
350	3.7638732268086e-05\\
400	2.58969390537176e-05\\
450	2.10042644264265e-05\\
500	1.60810340672412e-05\\
550	1.4936939438533e-05\\
600	1.12810830436536e-05\\
650	1.06256940693169e-05\\
700	9.29591094008269e-06\\
750	7.94980430512337e-06\\
800	7.23594336679583e-06\\
850	6.45913560633164e-06\\
900	5.78810771892764e-06\\
950	4.72019768570221e-06\\
1000	4.67790074879953e-06\\
1050	4.42561837008831e-06\\
1100	3.4264781345016e-06\\
1150	3.11402136900574e-06\\
1200	3.30427288718311e-06\\
1250	2.86335256546995e-06\\
1300	2.42380442083409e-06\\
1350	2.23531583990633e-06\\
1400	2.3167338812716e-06\\
1450	1.92779249402507e-06\\
1500	1.92741702371941e-06\\
1550	1.89930573797309e-06\\
1600	1.82752449219621e-06\\
1650	1.73562930270357e-06\\
1700	1.68802630232244e-06\\
1750	1.51213441583524e-06\\
1800	1.44101951202787e-06\\
1850	1.47449267247165e-06\\
1900	1.35238646243058e-06\\
1950	1.16804168398919e-06\\
2000	1.14408918950748e-06\\
2050	1.03126298982788e-06\\
2100	1.07543396832381e-06\\
2150	1.00369848577431e-06\\
2200	1.02024618973382e-06\\
2250	8.83469520982508e-07\\
2300	8.87968841901716e-07\\
2350	8.8708691415405e-07\\
2400	8.44806830394193e-07\\
2450	7.887145851219e-07\\
2500	7.59971862754684e-07\\
2550	6.94918363608735e-07\\
2600	6.63420355136418e-07\\
2650	6.01644135333382e-07\\
2700	6.35659444287138e-07\\
2750	6.58823468454524e-07\\
2800	5.69838643505703e-07\\
2850	5.67173508629303e-07\\
2900	4.61038715784029e-07\\
2950	5.92688752727595e-07\\
3000	5.25375761100662e-07\\
};\addlegendentry{suboptimal empirical variance};

\end{axis}
\end{tikzpicture}}
    \caption{For different alphabet sizes $n$, we plot in (a) the empirical mean of the maximum expected information gain (blue line) $\tfrac{1}{N} \sum_{i=1}^N \sup_{\lambda\in\Lambda}  I(p,\W_i^{(\lambda,V,n)})$, where $(\W_i^{(\lambda,V,n)})_{i=1}^N$ are independent channels and $N=1000$. The red line represents the empirical mean of the suboptimal expected information gain, that is given by $\tfrac{1}{N} \sum_{i=1}^N I(p, \W_i^{(\hat{\lambda},V,n)})$, where $\hat{\lambda}$ are the optimal parameters for the asymptotic capacity, derived in Proposition~\ref{prop:exp:design}.
(b) depicts the empirical variance of the maximum expected information gain (blue line) as well as the empirical variance of the suboptimal expected information gain (red line).   
    }
    \label{fig:plotExperimentMeanVar}
\end{figure}
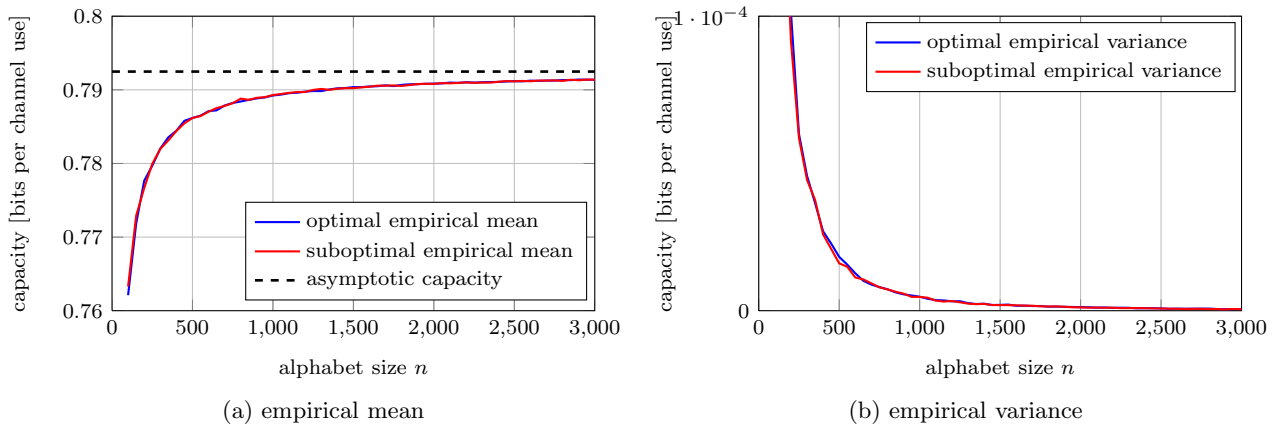

\section{Convergence Rate} \label{sec:finiteSize}
This section addresses how fast the capacity of a channel with the form introduced in Section~\ref{sec:main} converges to the asymptotic value predicted by Theorem~\ref{thm:main}. In addition, we consider a different model for the channel construction compared to Section~\ref{sec:main}.
Let $(V_{x,y})_{x\in [n],y\in[m]}$ be a sequence of independent nonnegative random variables such that the following assumption holds, where we use the notation $(x)_+^q = (0 \vee x )^q$.
\begin{myass}\label{ass:convergence:rate}
There exist positive numbers $K$ and $T$ such that for all $n,m\in\mathbb{N}$
\begin{enumerate}[(i)]
\item $\max \left\{\frac{1}{n} \sum\limits_{x\in[n]}\E{V_{x,y}^{2}}, \frac{1}{n} \sum\limits_{x\in[n]}\E{(V_{x,y}\log V_{x,y})^{2}} \right\} \leq K$ \ for all $y \in[m]$\label{ass:rate:i} \\
$\max \left\{\frac{1}{m} \sum\limits_{y\in[m]}\E{V_{x,y}^{2}}, \frac{1}{m} \sum\limits_{y\in[m]}\E{(V_{x,y}\log V_{x,y})^{2}} \right\} \leq K$ \ for all $x \in[n]$
\item $\max\left\{ \frac{1}{n} \sum\limits_{x\in[n]}\E{(V_{x,y})_{+}^{q}}, \frac{1}{n} \sum\limits_{x\in[n]}\E{(V_{x,y}\log V_{x,y})_{+}^{q}}\right\} \leq \frac{q!}{2}KT^{q-2}$ for all $q\geq 3, \, y \in[m]$\label{ass:rate:ii} \\
$\max\left\{ \frac{1}{m} \sum\limits_{y\in[m]}\E{(V_{x,y})_{+}^{q}}, \frac{1}{m} \sum\limits_{y\in[m]}\E{(V_{x,y}\log V_{x,y})_{+}^{q}}\right\} \leq \frac{q!}{2}KT^{q-2}$ for all $q\geq 3, \, x \in[n]$ 
\item $\frac{1}{m}\sum_{y\in[m]}\E{V_{x,y}}=\frac{1}{n}\sum_{x\in[n]}\E{V_{x,y}}$ for all $x\in [n]$ on the left hand side and all $y\in [m]$ on the right hand side \label{ass:rate:iii}
\item $V_{x,y}>0$ almost surely for all $x\in [n]$ and $y\in [m]$. \label{ass:rate:iv}
\end{enumerate}
\end{myass}
We denote by $\mu_{1,n}:=\tfrac{1}{m}\sum_{y\in[m]}\E{V_{x,y}}=\tfrac{1}{n}\sum_{x\in[n]}\E{V_{x,y}}$ and $\mu_{2,n}:=\tfrac{1}{m}\sum_{y\in[m]}\E{V_{x,y}\log V_{x,y}}$ and define the channel transition matrix $\W=(\W_{x,y})_{x\in [n], y\in[m]}$ by $\W_{x,y}=V_{x,y} / \sum_{y \in [m]}V_{x,y}$. Let $f : \Rp\times\Rp \to \Rp$ denote the function
\begin{equation} \label{eq:function:f}
f(t,n):= \exp \left(- \frac{nt^{2}}{ 2(K+Tt)} \right).
\end{equation}

The main difference between the random channel model considered in Section~\ref{sec:main} and the one in this section is that here we assume that the random variables $(V_{x,y})_{x\in[n],y \in [m]}$ are independent and such that Assumption~\ref{ass:convergence:rate} holds, whereas in Section~\ref{sec:main} we assume that the random variables $(V_{x,y})_{x\in[n],y \in [m]}$ are independent and \emph{identically distributed} and satisfy Assumption~\ref{ass:finite:moments}. Clearly Assumption~\ref{ass:convergence:rate} is stronger than Assumption~\ref{ass:finite:moments}, which allows us to state a rate of convergence. Note that Assumptions~\ref{ass:convergence:rate}\eqref{ass:rate:i} and \eqref{ass:rate:ii} are necessitated by the use of Bernstein's inequality (see Lemma~\ref{lem:bernstein}), Assumption~\ref{ass:convergence:rate}\eqref{ass:rate:iii} relates to the link with weakly symmetric channels and Assumption~\ref{ass:convergence:rate}\eqref{ass:rate:iv} significantly strengthens the positive mean assumption (Assumption~\ref{ass:finite:moments}).


\begin{mythm}[Rate of convergence] \label{thm:finiteSize}
Under Assumption~\ref{ass:convergence:rate}, the capacity of the DMC defined above satisfies for any $t\in\R_{>0}$
\begin{align*}
\Prob{\left|C(\Wn)-\left(\frac{\mu_{2,n}}{\mu_{1,n}}-\log \mu_{1,n} \right)\right|\geq t} & \leq \left( 2f(\alpha_{t/2},m)+f(\tfrac{t}{2L},m) \right) \vee \\
& \hspace{-13mm} \left( 2f(\alpha_{t/4},m)+f(\tfrac{t}{4L},m)+ f(\beta_{t/(2L)},n) + f(\beta_{t/(2L)},m) \right),
\end{align*}
with
\begin{align*}
\alpha_t &= \left \lbrace \begin{array}{ll}  
\tfrac{t \mu_{1,n}^2}{\mu_{1,n}(1+t)+\mu_{2,n}} & \textnormal{if } \mu_{1,n} + \mu_{2,n} \geq 0\\
\tfrac{t \mu_{1,n}^2}{\mu_{1,n}(1-t)+\mu_{2,n}} & \textnormal{otherwise}
\end{array} \right., \quad  \beta_t = \frac{t \mu_{1,n}}{2+t}, \quad  L=\frac{1}{a \ln 2} \quad  \textnormal{and} \\  
a&=\min\left\{ \frac{1}{m}\sum_{y=1}^{m}V_{x,y}, \mu_{1,n}, \sum_{k\in[n]}\frac{V_{k,y}}{\sum_{y\in[m]}V_{k,y}}, \frac{n}{m} \right\}.
\end{align*}
\end{mythm}

\begin{myremark}[Exponential convergence]
Note that $\mu_{1,n}$ is strictly larger than zero for any $n\in\mathbb{N}$ by Assumption~\ref{ass:convergence:rate}\eqref{ass:rate:iv}. Moreover, Assumption~\ref{ass:convergence:rate}\eqref{ass:rate:i} implies that there exists a constant $S$ such that $\mu_{2,n}\leq S$ for any $n\in\mathbb{N}$. Therefore, the parameters $\alpha_t$ and $\beta_t$ in Theorem~\ref{thm:finiteSize} can be bounded from below independently of $n$. Assumption~\ref{ass:convergence:rate}\eqref{ass:rate:iv} further ensures that the parameter $a$ can be bounded from below independently of $n$ and as such the parameter $L$ is bounded from above and below independently of $n$. Hence, Theorem~\ref{thm:finiteSize} clearly implies exponential convergence in $n$.
\end{myremark}

Assume that as $n\rightarrow \infty$,  $\mu_{1,n}$ and $\mu_{2,n}$ converge and denote the limits by $\bar{\mu}_{1}:=\lim_{n\to\infty}\mu_{1,n}$ and $\bar{\mu}_{2}:=\lim_{n\to\infty}\mu_{2,n}$
\begin{mycor}[Asymptotic capacity] \label{cor:asymptoticCapacity}
Under Assumptions~\ref{ass:input:output:size} and \ref{ass:convergence:rate}, for $n\to \infty$, the capacity $C(\Wn)$ of the DMC defined above converges to $\tfrac{\bar{\mu}_{2}}{\bar{\mu}_{1}} - \log \bar{\mu}_{1}$ in probability.
\end{mycor}
\begin{proof}
Follows directly from Theorem~\ref{thm:finiteSize}.
\end{proof}
Since the exponential concentration provided in Theorem~\ref{thm:finiteSize} is summable, a direct application of the Borel-Cantelli Lemma \cite[Theorem~2.3.1]{durrett_book} allows us to improve Corollary~\ref{cor:asymptoticCapacity} to almost sure convergence.

\subsection{Proof of Theorem~\ref{thm:finiteSize}} \label{sec:proof:rate}

The structure of the proof is such that we prove separately convergence rates for the lower and upper bounds of Section~\ref{sec:main} (Propositions~\ref{prop:UBfinite} and \ref{prop:LBfinite}) respectively. The claim follows since the capacity is forced to be between the upper and lower bounds, hence converges at the worst among the two rates.

\begin{myprop} \label{prop:UBfinite}
A random channel $\W$ as introduced in this section with $C^{(\lambda=0)}_{\UB}(\W)$ given in \eqref{eq:defdUB} satisfies
\begin{equation*}
\Prob{\left|C^{(\lambda=0)}_{\UB}(\W)-\left(\frac{\mu_{2,n}}{\mu_{1,n}}-\log \mu_{1,n} \right)\right|\geq t}\leq  2f(\alpha_{t/2},n)+f(\tfrac{t}{2L},n) 
\end{equation*}
with 
\begin{equation*}
\alpha_t = \left \lbrace \begin{array}{ll}  
\tfrac{t \mu_{1,n}^2}{\mu_{1,n}(1+t)+\mu_{2,n}} & \textnormal{if } \mu_{1,n} + \mu_{2,n} \geq 0\\
\tfrac{t \mu_{1,n}^2}{\mu_{1,n}(1-t)+\mu_{2,n}} & \textnormal{otherwise}
\end{array} \right. \quad \textnormal{where} \quad L=\tfrac{1}{a \ln 2} \ \textnormal{and } a=\min\left\{ \tfrac{1}{m}\sum_{y=1}^{m}V_{x,y}, \mu_{1,n} \right\}.
\end{equation*}
\end{myprop}
\begin{proof}
See Appendix~\ref{app:proof:rate}.
\end{proof}

\begin{myprop} \label{prop:LBfinite}
A random channel $\W$ of the form introduced in this section with $C_{\LB}^{(p \sim \mathcal{U})}(\W)$ given in \eqref{eq:dav} satisfies
\begin{align*}
&\Prob{\left|C^{(p\sim\mathcal{U})}_{\LB}-\left(\frac{\mu_{2,n}}{\mu_{1,n}}-\log \mu_{1,n} \right)\right|\geq t} \leq \\
 &\hspace{30mm}2f(\alpha_{t/4},m)+f(\tfrac{t}{4L},m)+f(\beta_{t/(2L)},n)+f(\beta_{t/(2L)},m)
\end{align*}
with 
\begin{align*}
\alpha_t &= \left \lbrace \begin{array}{ll}  
\tfrac{t \mu_{1,n}^2}{\mu_{1,n}(1+t)+\mu_{2,n}} & \textnormal{if } \mu_{1,n} + \mu_{2,n} \geq 0\\
\tfrac{t \mu_{1,n}^2}{\mu_{1,n}(1-t)+\mu_{2,n}} & \textnormal{otherwise}
\end{array} \right., \quad \beta_t = \frac{t \mu_{1,n}}{2+t}, \quad L=\tfrac{1}{a \ln 2} \quad \textnormal{and}\\
a&=\min\left\{ \sum_{k\in[n]}\frac{V_{k,y}}{\sum_{y\in[m]}V_{k,y}}, \frac{n}{m} \right\}.
\end{align*}
\end{myprop}
\begin{proof}
See Appendix~\ref{app:proof:rate}.
\end{proof}

\begin{proof}[Proof of Theorem~\ref{thm:finiteSize}]
Theorem~\ref{thm:finiteSize} follows directly from Propositions~\ref{prop:UBfinite} and \ref{prop:LBfinite} as by definition $C^{(p\sim\mathcal{U})}_{\LB}(\W) \leq C(\W) \leq C^{(\lambda=0)}_{\UB}(\W)$ almost surely. 
\end{proof}

\section{Conclusion and Discussion} \label{sec:discussion}
In this article we studied the capacity of discrete memoryless channels whose channel transition matrix consists of entries that are nonnegative i.i.d.\ random variables $V$ before being normalized. It was shown that under some mild assumptions on the distribution of the random variables, the capacity of such a channel as the dimension goes to infinity converges to the \emph{asymptotic capacity} given by
$\tfrac{\mu_2}{\mu_1}-\log \mu_1$ almost surely and in $\Lp{2}$, where $\mu_{1}:=\E{V}$ and $\mu_{2}:=\E{V\log V}$. Interestingly, for some distributions, e.g., the uniform and exponential distribution, the asymptotic capacity is a constant. Furthermore, we have shown that the capacity of these random channels converges exponentially to its asymptotic value in probability.
Finally, we provided an interpretation of the asymptotic capacity as an upper bound to the maximum expected information gain in the context of Bayesian optimal experiment design.

For future work we aim to investigate if the asymptotic capacity of a random channel determined by Theorem~\ref{thm:main} has an operational meaning in other scenarios, e.g., in the setup of fading channels or in Bayesian estimation.
Furthermore, it would be interesting to study the variance of the capacity of such random channels and its decay rate.


\appendix

\section{Measurability of the capacity}\label{app:measurability}
We show that the the capacity $C(\Wn)$ of such a (random) DMC as well as the optimal input distribution are random variables.
\begin{mylem}[Measurability] \label{lem:measurability}
For a channel constructed as explained above the mapping $C:\mathcal{M}_{n,m}\to\Rp$ given by $C(\Wn)=\max_{p\in\Delta_{n}}\I{p}{\Wn}$ is measurable.
Furthermore, the (set-valued) mapping $p^{\star}:\mathcal{M}_{n,m}\rightrightarrows \Delta_{n}$, $p^{\star}(\Wn)=\arg \max_{p\in\Delta_{n}} \I{p}{\Wn}$, describing the optimal input distribution, is measurable.
\end{mylem}
\begin{proof}
Note that we have
\begin{equation*}
C(\Wn)=\max_{p\in\R^{n}}\{I(p,\Wn) + \hat{\delta}_{\Delta_{n}}(p)\}, \quad \textnormal{where} \quad \hat{\delta}_{\Delta_{n}}(p)=\left\{ \begin{array}{ll} 0, &\text{if }p\in\Delta_{n} \\ -\infty, &\text{otherwise.} \end{array} \right.
\end{equation*}
Since the mapping $p\mapsto \I{p}{\Wn}$ is concave and continuous for almost any $\Wn$, $I$ is a normal integrand \cite[Proposition~14.39]{ref:Rockafellar-97}.
Then, as shown in \cite[Example~14.32]{ref:Rockafellar-97}, $I(p,\Wn) + \hat{\delta}_{\Delta_{n}}(p)$ is a normal integrand and as such the measurability of the mappings $\Wn\mapsto C(\Wn)$ 
and $\Wn\mapsto p^{\star}(\Wn)$
follows by \cite[Theorem~14.37]{ref:Rockafellar-97}; see \cite[Definition~14.1]{ref:Rockafellar-97} for a definition of measurability of a set-valued mapping.
\end{proof}

\section{Additional examples}\label{app:additional:examples}

\begin{myex}[Uniform distribution] \label{ex:uniform}
Consider a DMC as defined above, where the elements $V_{x,y}$ are uniformly distributed with support $[0,A]$ for some $A> 0$. It can be seen by the homogeneity property in Remark~\ref{rmk:properties:asym:cap}, that the asymptotic capacity cannot depend on $A$. More precisely, for $n \to \infty$ the capacity converges to $1-\frac{1}{2 \ln 2}$ almost surely and in $\Lp{2}$, since for $V_{x,y}\sim\mathcal{U}([0,A])$, we have $\mu_1=\E{V_{x,y}}=\tfrac{A}{2}$ and $\mu_2=\E{V_{x,y}\log V_{x,y}}=\tfrac{A(2 \ln(A) - 1)}{4 \ln 2}$. 
\end{myex}

\begin{myex}[Gamma distribution] \label{ex:gamma}
Consider a DMC as defined in Section~\ref{sec:main} using a gamma distribution with shape and scale parameter $k>0$ and $\theta>0$. For $n \to \infty$ its capacity converges to $\tfrac{\psi(1+k)}{\ln 2}-\log k$ almost surely and in $\Lp{2}$, where $\psi(\cdot)$ denotes the \emph{digamma} function. This is a direct consequence of Theorem~\ref{thm:main}, since for $V_{x,y} \sim \Gamma(k,\theta)$ we have $\mu_1=\E{V_{x,y}}=k \theta$ and $\mu_2=\E{V_{x,y}\log V_{x,y}}=k \theta \tfrac{\ln\psi(1+k) + \ln \theta}{\ln 2}$.
We note that $\alpha V_{x,y} \sim \Gamma(k,\alpha\theta)$ for positive $\alpha$, which by the homogeneity property (cf.~Remark~\ref{rmk:properties:asym:cap}) implies that the asymptotic capacity cannot depend on $\theta$.
\end{myex}

\begin{myex}[Chi-squared distribution] \label{ex:chiSquared}
Consider a DMC as defined in Section~\ref{sec:main} using a chi-squared distribution with degrees of freedom $k\in \mathbb{N}$. For $n \to \infty$ its capacity converges to $1 + \tfrac{1}{\ln 2} \psi(1+\tfrac{k}{2})-\log k$ almost surely and in $\Lp{2}$, where $\psi(\cdot)$ denotes the \emph{digamma} function. This is a direct consequence of Theorem~\ref{thm:main}, since for $V_{x,y} \sim \chi^2(k)$ we have $\mu_1=\E{V_{x,y}}=k$ and $\mu_2=\E{V_{x,y}\log V_{x,y}}=k + \tfrac{k}{\ln 2}\psi(1+\tfrac{k}{2})$.
\end{myex}

\begin{myex}[Beta distribution] \label{ex:Beta}
Consider a DMC as defined in Section~\ref{sec:main} using a beta distribution with shape parameters $\alpha,\beta >0$. For $n \to \infty$ its capacity converges to $\tfrac{H_{\alpha}-H_{\alpha+\beta}}{\ln 2} - \log \tfrac{\alpha}{\alpha+\beta}$ almost surely and in $\Lp{2}$, where $H_n$ denotes the $n$-th harmonic number. This is a direct consequence of Theorem~\ref{thm:main}, using that for $V_{x,y} \sim \textnormal{beta}(\alpha,\beta)$ we have $\mu_1=\E{V_{x,y}}=\tfrac{\alpha}{\alpha+\beta}$ and $\mu_2=\E{V_{x,y}\log V_{x,y}}=\tfrac{\alpha}{(\alpha+\beta)\ln 2}(H_{\alpha}-H_{\alpha+\beta})$.
\end{myex}
\section{Two technical lemmas}

\begin{mylem} \label{lemma:ImpLimit}
Let $ X_1,X_2,\ldots, X_n$ be i.i.d.\ nonnegative random variables with $\E{X_i}=:\mu_1 > 0$, $\E{X_i \log X_i}=:\mu_2$ and $\E{(X_i \log X_i)^{2}}< \infty$. Let $Y_i = \frac{X_i}{\sum_{j=1}^n X_j}$ then as $n\to \infty$, $\sum_{i=1}^n Y_i \log Y_i + \log n \to \tfrac{\mu_2}{\mu_1} - \log \mu_1$ almost surely and in $\Lp{2}$.
\end{mylem}
\begin{proof}
Let $\xi_n := \tfrac{1}{n} \sum_{j=1}^n X_j$ and $Z_i:=X_i \log X_i$. We then can write
\begin{align}
\sum_{i=1}^n Y_i \log Y_i +Ê\log n &= \sum_{i=1}^n \frac{X_i}{\sum_{j=1}^n X_j} \log\left(\frac{X_i}{\sum_{j=1}^n X_j} \right) + \log n \nonumber \\
&=\frac{1}{n} \sum_{i=1}^n \frac{X_i}{\xi_n} \log \left( \frac{X_i}{n \, \xi_n} \right) + \log n  \nonumber\\
&= \frac{1}{\xi_n} \left( \frac{1}{n} \sum_{i=1}^n Z_i - \log(n\, \xi_n) \frac{1}{n} \sum_{i=1}^n X_i \right) + \log n  \nonumber \\
&= \frac{1}{\xi_n} \frac{1}{n} \sum_{i=1}^n  Z_i - \log \xi_n. \label{eq:apAA}
\end{align}
Note that $\E{(X_{i}\log X_{i})^{2}}<\infty$ implies that $X_{i}$ has a finite second moment.
Using the strong law of large numbers \cite[Theorem~2.4.1]{durrett_book}, it follows that for $n\to \infty$, $\xi_n \to \mu_1$ almost surely and $\tfrac{1}{n} \sum_{i=1}^n Z_i \to \mu_2$ almost surely. The convergence in $\Lp{2}$ follows by using the $\Lp{2}$-weak law \cite[Theorem 2.2.3.]{durrett_book} instead of the strong law of large numbers.
\end{proof}


\begin{mylem} \label{lemma:normalization}
Let $\{X_{i,j}\}_{i \in [n],\, j\in[m]}$ be i.i.d.\ random variables taking values on $\Rp$, with $\E{ X_{i,j}}=:\mu> 0$ and $\E{X_{i,j}^{2}}<\infty$. If $Y_{i,j}=\frac{ X_{i,j}}{\sum_{k=1}^m X_{i,k}}$ then for $n\to \infty$, where $m:=\ceil{\gamma n}$ for some $\gamma\in\Rsp$, $\sum_{i=1}^n Y_{i,j} \to \tfrac{1}{\gamma}$ almost surely and in $\Lp{2}$ for every $j \in [m]$. 
\end{mylem}
\begin{proof}
By assumption we can write 
\begin{align}
\sum_{i=1}^n Y_{i,j} &= \sum_{i=1}^n \frac{ X_{i,j}}{\sum_{k=1}^m  X_{i,k}} \nonumber\\
&= \frac{1}{\gamma n + \varepsilon_{n}} \sum_{i=1}^n \frac{ X_{i,j}}{\frac{1}{m}\sum_{k=1}^m  X_{i,k}}, \label{eq:idk}
\end{align}
where $\varepsilon_{n}:= \ceil{\gamma n}-\gamma n\in[0,1)$ for all $n\in\mathbb{N}$.
We can bound \eqref{eq:idk} from above as
\begin{align}
\frac{1}{\gamma n + \varepsilon_{n}} \sum_{i=1}^n \frac{ X_{i,j}}{\frac{1}{m}\sum_{k=1}^m  X_{i,k}} &\leq \frac{1}{\gamma n + \varepsilon_{n}} \sum_{i=1}^n \frac{ X_{i,j}}{\frac{1}{m} \min \limits_{\ell \in [n]} \sum_{k=1}^m  X_{\ell,k}} \nonumber \\
 &= \frac{1}{\gamma} \frac{ \frac{1}{1+\frac{\varepsilon_{n}}{\gamma n}}\frac{1}{n} \sum_{i=1}^n X_{i,j}}{\frac{1}{m} \sum_{k=1}^m X_{\underline{\ell},k}}, \label{eq:ub}
\end{align}
where $\underline{\ell} \in \arg \min _{\ell\in [n]} \sum_{k=1}^m  X_{\ell,k}$. By the strong law of large numbers \cite[Theorem 2.4.1]{durrett_book} respectively the $\Lp{2}$-weak law \cite[Theorem 2.2.3.]{durrett_book}, the right hand side of \eqref{eq:ub} converges to $1/\gamma$ almost surely, respectively in $\Lp{2}$ for $n\to \infty$.
We can also bound \eqref{eq:idk} from below by
\begin{align}
 \frac{1}{\gamma n + \varepsilon_{n}} \sum_{i=1}^n \frac{ X_{i,j}}{\frac{1}{m}\sum_{k=1}^m  X_{i,k}} &\geq \frac{1}{\gamma n + \varepsilon_{n}} \sum_{i=1}^n \frac{ X_{i,j}}{\frac{1}{m} \max \limits_{\ell \in [n]} \sum_{k=1}^m  X_{\ell,k}}  \nonumber\\
 &= \frac{1}{\gamma} \frac{ \frac{1}{1+\frac{\varepsilon_{n}}{\gamma n}}\frac{1}{n} \sum_{i=1}^n X_{i,j}}{\frac{1}{m} \sum_{k=1}^m X_{\overline{\ell},k}}, \label{eq:lb}
\end{align}
where $\overline{\ell} \in\arg \max _{\ell\in [n]} \sum_{k=1}^n  X_{\ell,k}$. 
Again the strong law of large numbers and the $\Lp{2}$-weak law ensure that the right hand side of \eqref{eq:lb} converges to $1/\gamma$ almost surely and in $\Lp{2}$ for $n \to \infty$. This then implies that \eqref{eq:idk} converges to $1/\gamma$ almost surely and in $\Lp{2}$ for $n\to \infty$ which proves the assertion. 
Note that the law of large numbers in \eqref{eq:lb} works there because the number of sequences is linear in the sequence length.
\end{proof}

\section{Proof of Theorem~\ref{thm:finiteSize}} \label{app:proof:rate}
As already sketched in Section~\ref{sec:proof:rate}, the proof of Theorem~\ref{thm:finiteSize} is implied by convergence rate statements for the upper and lower bounds of Section~\ref{sec:main} (Propositions~\ref{prop:UBfinite} and \ref{prop:LBfinite}) respectively. This appendix provides a proof for the mentioned two propositions. To prove Proposition~\ref{prop:UBfinite} we need a few preparatory lemmas.
\begin{mylem}[{Bernstein's inequality \cite[Corollary~2.11]{ref:boucheron-13}}] \label{lem:bernstein}
Let $X_{1},\hdots, X_{n}$ be independent real-valued random variables. Assume that there exist positive numbers $K$ and $T$ such that $\sum_{i=1}^{n}\E{X_{i}^{2}}\leq K$ and
\begin{equation*}
\sum_{i=1}^{n}\E{(X_{i})_{+}^{q}}\leq \frac{q!}{2}KT^{q-2} \quad \text{for all integers } q\geq 3.
\end{equation*}
If $S=\tfrac{1}{n}\sum_{i=1}^{n}(X_{i}-\E{X_{i}})$, then for all $t>0$
\begin{equation*}
\Prob{S\geq t} \leq \exp\left( \frac{-n^{2}t^{2}}{2(K+Tnt)} \right).
\end{equation*}

\end{mylem}

\begin{mylem} \label{lem:basic}
Let $X$ and $Y$ be two random variables, $\eta_1$ and $\eta_2$ be two real constants such that $\Prob{|X-\eta_1|\geq t} \leq g_1(t)$ and $\Prob{|Y-\eta_2|\geq t} \leq g_2(t)$ for two functions $g_i : \R \to \Rp$, $i\in \{1,2 \}$. Then, $\Prob{\left|X+Y-\eta_1-\eta_2 \right|\geq t} \leq g_1(\tfrac{t}{2}) + g_2(\tfrac{t}{2})$.
\end{mylem}
\begin{proof}
Consider the following four events $A_t:=\{|X-\eta_1|\geq t \}$, $B_t:=\{|Y-\eta_2|\geq t \}$, $C_t:=\{|X+Y-\eta_1 - \eta_2|\geq t \}$ and $D_t:=\{|X-\eta_1|+|Y-\eta_2|\geq tÊ\}$. Using the triangle inequality and the union bound we find
\begin{equation*}
\Prob{C_t} \leq \Prob{D_t} \leq \Prob{A_{t/2} \cup B_{t/2}} \leq \Prob{A_{t/2}} + \Prob{B_{t/2}} \leq g_1(\tfrac{t}{2}) + g_2(\tfrac{t}{2}).
\end{equation*}
\end{proof}
\begin{mylem} \label{lem:division}
Let $X$, $Y$ be random variables and $\eta_1 \in \R$, $\eta_2 \in \Rsp$ constants such that $\Prob{|X-\eta_1|\geq t}\leq g_1(t)$ and $\Prob{|Y-\eta_2|\geq t}\leq g_2(t)$, for two functions $g_i : \R \to \Rp$, $i \in \{1,2\}$ and $t\in \Rp$. Then 
\begin{equation*}
\Prob{\left|\frac{X}{Y}-\frac{\eta_1}{\eta_2} \right|\geq t} \leq g_1(\alpha_t) + g_2(\alpha_t)  \quad \textnormal{with} \quad \alpha_t = \left \lbrace \begin{array}{ll}  
\tfrac{t \eta_2^2}{\eta_2(1+t)+\eta_1} & \textnormal{if } \eta_1 + \eta_2 \geq 0\\
\tfrac{t \eta_2^2}{\eta_2(1-t)+\eta_1} & \textnormal{otherwise}.
\end{array} \right.
\end{equation*}
\end{mylem}
\begin{proof}
Consider the three events $A_t:=\{|X-\eta_1|\geq t \}$, $B_t:=\{|Y-\eta_2|\geq t \}$ and $C_t:=\{|\tfrac{X}{Y}-\tfrac{\eta_1}{\eta_2}|\geq t \}$. We first show that $A_t^{\setC} \cap B_t^{\setC} \subseteq C_{\gamma_t}^{\setC}$, with $\gamma_t:=\tfrac{t(\eta_1 +\eta_2)}{\eta_2(\eta_2-t)}$ for $t\neq \eta_{2}$ and $\eta_1 + \eta_2 \geq 0$.
Given $A_t^{\setC} \cap B_t^{\setC}$  it follows that $X\in[\eta_1-t,\eta_1+t]$ and $Y\in[\eta_2-t,\eta_2+t]$. Given $A_t^{\setC} \cap B_t^{\setC}$, two possible extreme values of $|\tfrac{X}{Y}-\tfrac{\eta_1}{\eta_2}|$ are
\begin{align} 
\left| \frac{\eta_1+t}{\eta_2-t}-\frac{\eta_1}{\eta_2} \right| &= \left| \frac{t(\eta_2+\eta_1)}{\eta_2(\eta_2-t)} \right| =:\gamma_t \quad  \textnormal{ for $t\neq \eta_{2}$ \,\,\,\, and} \label{eq:smaller} \\
\left| \frac{\eta_1}{\eta_2}-\frac{\eta_1-t}{\eta_2+t} \right| &= \left| \frac{t(\eta_1+\eta_2)}{\eta_2(\eta_2+t)} \right| =:\gamma'_t, \label{eq:beta}
\end{align} 
where it is immediate that for $\eta_1 + \eta_2 \geq 0$, we have $\gamma_t \geq \gamma'_t$. We thus have $|\tfrac{X}{Y}-\tfrac{\eta_1}{\eta_2}|\leq \gamma_t$, which implies by definition that $A_t^{\setC} \cap B_t^{\setC} \subseteq C_{\gamma_t}^{\setC}$ and thus
\begin{equation}
\Prob{C_{\gamma_t}^{\setC}} \geq \Prob{A_t^{\setC} \cap B_t^{\setC}}. \label{eq:firststep}
\end{equation}
Using \eqref{eq:firststep}, de Morgan's law and the union bound we find
\begin{align}
\Prob{C_{\gamma_t}} &= 1-\Prob{C_{\gamma_t}^{\setC}} \leq 1- \Prob{A_t^{\setC} \cap B_t^{\setC}} =\Prob{A_t \cup B_t}  \leq \Prob{A_t} + \Prob{B_t}.  \label{eq:final}
\end{align}
Solving $\gamma_t=\left| \tfrac{t(\eta_1 +\eta_2)}{\eta_2(\eta_2-t)}\right|$ for $t\neq\eta_{2}$ and inserting it into \eqref{eq:final} proves the assertion for $\eta_1 + \eta_2 \geq 0$. If $\eta_1 + \eta_2 < 0$, we have $\gamma'_t \geq \gamma_t$. Following the same lines as above, i.e., solving $\gamma'_t=\left| \tfrac{t(\eta_1 +\eta_2)}{\eta_2(\eta_2+t)}\right|$ for $t$ and inserting it into \eqref{eq:final} proves the assertion for $\eta_1 + \eta_2 < 0$.
\end{proof}
\begin{mylem} \label{lem:lipschitz}
Let $X$ be a random variable and $\eta$ be a constant such that $X,\eta \in [\alpha,\infty)$ for $\alpha >0$ and $\Prob{|X-\eta|\geq t}\leq g(t)$ for some function $g:\Rp \to \Rsp$. Then $\Prob{|\log X - \log \eta|\geq t} \leq g(\tfrac{t}{L})$ with $L=\tfrac{1}{\alpha \ln(2)}$.
\end{mylem}
\begin{proof}
The function $h:[\alpha,\infty) \to \R$ for $\alpha >0$ that maps $x \mapsto \log x$ is known to be Lipschitz continuous with Lipschitz constant $L=\tfrac{1}{\alpha \ln 2}$. By definition of Lipschitz continuity we obtain
\begin{align*}
\Prob{|\log X - \log \eta|\geq t}  \leq \Prob{|X-\eta|\geq \frac{t}{L}} \leq g\left(\frac{t}{L}\right).
\end{align*}
\end{proof}

\begin{proof}[Proof of Proposition~\ref{prop:UBfinite}]
Let $\{V_{x,y}\}_{x \in [n],y \in [m]}$, $\mu_{1,n}$ and $\mu_{2,n}$ as defined above and let $Z_{x,y}:=V_{x,y} \log V_{x,y}$.
According to Assumption~\ref{ass:convergence:rate} and Bernstein's inequality (Lemma~\ref{lem:bernstein}),
\begin{equation}
\Prob{\left|\frac{1}{m} \sum_{y=1}^m V_{x,y} - \mu_{1,n} \right| \geq t} \leq \exp\left( \frac{-m^{2}t^{2}}{2(K+Tmt)} \right) =f(t,m) \quad \forall x \in [n], \label{eq:stepi}
\end{equation}
and
\begin{equation}
\Prob{\left|\frac{1}{m} \sum_{y=1}^m Z_{x,y} - \mu_{2,n} \right| \geq t} \leq \exp\left( \frac{-m^{2}t^{2}}{2(K+Tmt)} \right)=f(t,m) \quad \forall x \in [n], \label{eq:stepii}
\end{equation}
According to Lemma~\ref{lem:division}, \eqref{eq:stepi} and \eqref{eq:stepii} imply that
\begin{equation}
\Prob{\left|\frac{\frac{1}{m}\sum_{y=1}^m Z_{x,y}}{\frac{1}{m}\sum_{y=1}^m V_{x,y}}-\frac{\mu_{2,n}}{\mu_{1,n}} \right|\geq t} \leq 2f(\alpha_t,m) \quad \forall x\in [n], \label{eq:stepiii}
\end{equation}
for
\begin{equation*}
\alpha_t = \left \lbrace \begin{array}{ll}  
\tfrac{t \mu_{1,n}^2}{\mu_{1,n}(1+t)+\mu_{2,n}} & \textnormal{if } \mu_{1,n} + \mu_{2,n} \geq 0\\
\tfrac{t \mu_{1,n}^2}{\mu_{1,n}(1-t)+\mu_{2,n}} & \textnormal{otherwise.}
\end{array} \right.
\end{equation*}
Lemma~\ref{lem:lipschitz} together with \eqref{eq:stepi} gives
\begin{align}
\Prob{\left|\log\left(\frac{1}{m} \sum_{y=1}^m V_{x,y} \right) - \log \mu_{1,n}\right|\geq t} \leq f(\tfrac{t}{L},m) \quad \forall x \in [n], \label{eq:stepiiii}
\end{align}
with $L=\tfrac{1}{a \ln 2}$ and $a=\min\left\{ \tfrac{1}{m}\sum_{y=1}^{m}V_{x,y}, \mu_{1,n} \right\}$.
Finally, using the definition of $C^{(\lambda=0)}_{\UB}(\W)$ given in \eqref{eq:defdUB} we find
\begin{align}
&\Prob{\left|C^{(\lambda=0)}_{\UB}(\W)-\left(\frac{\mu_{2,n}}{\mu_{1,n}}-\log \mu_{1,n}\right) \right| \geq t} \nonumber \\
&\hspace{10mm}= \Prob{\left| \max_{x\in [n]}\left \lbrace \sum_{y=1}^m \W_{x,y} \log \W_{x,y} \right \rbrace + \log m - \left(\frac{\mu_{2,n}}{\mu_{1,n}}-\log \mu_{1,n} \right) \right| \geq t} \nonumber\\
&\hspace{10mm}=\Prob{\left| \max_{x\in [n]}\left \lbrace \sum_{y=1}^m \frac{V_{x,y}}{\sum_{k=1}^m V_{x,k}} \log \left( \frac{V_{x,y}}{\sum_{\ell=1}^m V_{x,\ell}}\right) \right \rbrace + \log m - \left(\frac{\mu_{2,n}}{\mu_{1,n}}-\log \mu_{1,n} \right) \right| \geq t} \nonumber\\
&\hspace{10mm}=\Prob{\left|  \sum_{y=1}^m \frac{V_{\overline{x},y}}{\sum_{k=1}^m V_{\overline{x},k}} \log \left( \frac{V_{\overline{x},y}}{\sum_{\ell=1}^m V_{\overline{x},\ell}}\right)  + \log m - \left(\frac{\mu_{2,n}}{\mu_{1,n}}-\log \mu_{1,n} \right) \right| \geq t} \label{eq:maxxx} \\
&\hspace{10mm}= \Prob{\left|\frac{\frac{1}{m}\sum_{y=1}^m Z_{\overline{x},y}}{\frac{1}{m}\sum_{y=1}^m V_{\overline{x},y}} - \frac{\mu_{2,n}}{\mu_{1,n}}- \left(\log\Biggl(\frac{1}{m} \sum_{y=1}^m V_{\overline{x},y} \Biggr)- \log \left(\mu_{1,n}\right)\right)\right|\geq t} \label{eq:basicc}\\
&\hspace{10mm} \leq 2f(\alpha_{t/2},m)+ f(\tfrac{t}{2L},m), \label{eq:ddavv}
\end{align}
where in \eqref{eq:maxxx} $\overline{x}$ denotes the $x\in[n]$ that achieves the maximum. Equation \eqref{eq:basicc} follows by recalling that $Z_{x,y}:= V_{x,y} \log V_{x,y}$.
The inequality finally uses \eqref{eq:stepiii}, \eqref{eq:stepiiii} and Lemma~\ref{lem:basic}.
\end{proof}

We next derive a few preparatory lemmas that are used to prove Proposition~\ref{prop:LBfinite}.
\begin{mylem} \label{lem:LBi}
Using the notation introduced above, for all $y \in [m]$,
\begin{equation*}
\Prob{\left|\sum_{k \in [n]} \W_{k,y} -\frac{n}{m} \right|\geq t} \leq f(\beta_t,n)+f(\beta_t,m), \quad \textnormal{with} \quad \beta_t:=\frac{t \mu_{1,n}}{2+t}.
\end{equation*}
\end{mylem}
\begin{proof}
Fix an arbitrary $y\in [n]$ and define
\begin{equation*}
U_y:=\sum_{k \in [n]} \W_{k,y}=\frac{n}{m}\frac{1}{n} \sum_{k \in [n]} \frac{V_{k,y}}{\frac{1}{n}\sum_{\ell \in [n]}V_{k,\ell}}.
\end{equation*}
With $\overline{k}\in \arg \max_{k \in [n]} \sum_{\ell \in [m]} V_{k,\ell}$ and $\underline{k}\in \arg \min_{k \in [n]} \sum_{\ell \in [m]} V_{k,\ell}$ we can bound $U_y$ from below and above by
\begin{align*}
U_{y,\LB}=\frac{n}{m}\frac{1}{n} \sum_{k \in [n]} \frac{V_{k,y}}{\frac{1}{m}\sum_{\ell \in [m]}V_{\overline{k},\ell}} \quad \textnormal{and} \quad 
U_{y,\UB}=\frac{n}{m}\frac{1}{n} \sum_{k \in [n]} \frac{V_{k,y}}{\frac{1}{m}\sum_{\ell \in [m]}V_{\underline{k},\ell}}.
\end{align*}
Lemma~\ref{lem:bernstein} and Lemma~\ref{lem:division} give
\begin{align*}
\Prob{\left|U_{y,\LB}-\frac{n}{m} \right|\geq t} & = \Prob{\left|\frac{n}{m}\frac{\frac{1}{n}\sum_{k\in[n]}V_{k,y}}{\frac{1}{m}\sum_{\ell \in [m]}V_{\overline{k},\ell}}-\frac{n}{m} \right|\geq t}\\
&\leq  f(\beta_t,n)+f(\beta_t,m),
\end{align*}
for $f(\cdot,\cdot)$ and $\beta_t$ as defined in \eqref{eq:function:f} and the theorem. The same argument can be obtained to bound $\Prob{\left|U_{y,\UB}-\tfrac{n}{m} \right|\geq t}$ which then proves the assertion.
\end{proof}
\begin{mylem} \label{lem:LBii}
Using the notation introduced above, for every $y \in [n]$, 
\begin{equation*}
\Prob{\left|\log \sum_{k \in [n]} \W_{k,y}-\log \frac{n}{m} \right|\geq t} \leq f\left(\beta_{t/L},n \right)+f\left(\beta_{t/L},m \right)\
\end{equation*}
$\textnormal{with} \quad \beta_t:=\frac{t \mu_{1,n}}{2+t}, \quad L = \tfrac{1}{a \ln 2} \quad \textnormal{and} \quad a=\min\left\{ \sum_{k\in[n]}\frac{V_{k,y}}{\sum_{y\in[m]}V_{k,y}}, \frac{n}{m} \right\}.$
\end{mylem}
\begin{proof}
Follows directly from Lemmas \ref{lem:lipschitz} and \ref{lem:LBi}.
\end{proof}
\begin{mylem} \label{lem:LBiii}
Using the notation introduced above 
\begin{equation*}
\Prob{\left|\frac{1}{n} \sum_{x \in [n],y \in [m]} \W_{x,y} \log \sum_{k \in [n]} \W_{k,y}-\log \frac{n}{m}\right|\geq t} \leq f\left(\beta_{t/L},n\right)+f\left(\beta_{t/L},m\right),
\end{equation*}
with $\beta_t:=\tfrac{t \mu_{1,n}}{2+t} \quad L = \tfrac{1}{a \ln 2} \quad \textnormal{and} \quad a=\min\left\{ \sum_{k\in[n]}\frac{V_{k,y}}{\sum_{y\in[m]}V_{k,y}}, \frac{n}{m} \right\}.$
\end{mylem}
\begin{proof}
For an arbitrary $y \in [m]$, define the events
\begin{align*}
A_t:=\left \lbrace \left|\log \sum_{k \in [n]} \W_{k,y}-\log \frac{n}{m} \right| \geq t\right \rbrace \quad \textnormal{and} \quad B_t:=\left \lbrace \left|\frac{1}{n}\!\! \sum_{x\in[n],y \in [m]}\!\!\!\!\!\! \W_{x,y} \log\sum_{k \in [n]} \W_{k,y} -\log \frac{n}{m}\right| \geq tÊ\right \rbrace.
\end{align*}
Suppose that for all $y \in [m]$, $|\log \sum_{k \in [n]} \W_{k,y}-\log \tfrac{n}{m}|\leq t$, this implies that 
\begin{equation*}
\left|\frac{1}{n} \sum_{x,y \in [n]} \W_{x,y} \log\sum_{k \in [n]} \W_{k,y} \right| \leq t,
\end{equation*}
where we used that $\W_{x,y}\geq 0$ for all $x\in[n]$, $y\in[m]$ and that for all $x\in[n]$, $\sum_{y \in [m]} \W_{x,y}=1$.
Thus $\Prob{A_t^{\setC}} \leq \Prob{B_t^{\setC}}$ or equivalently $\Prob{B_t} \leq \Prob{A_t}$, which together with Lemma~\ref{lem:LBii} proves the assertion.
\end{proof}

\begin{proof}[Proof of Proposition~\ref{prop:LBfinite}]
Let $\underline{x}\in\arg \min_{x\in[n]} \sum_{x,y\in[n]}\W_{x,y}\log \W_{x,y}$.
By definition of $C_{\LB}^{(p \sim \mathcal{U})}(\W)$ given in \eqref{eq:dav} we have
\begin{align*}
&\Prob{\left|C^{(p\sim\mathcal{U})}_{\LB}(\W)-\left(\frac{\mu_{2,n}}{\mu_{1,n}}-\log \mu_{1,n} \right)\right|\geq t}  \nonumber\\
& \hspace{6mm} =\Prob{\left| \log n + \frac{1}{n}\!\! \sum_{x\in [n],y \in [m]}\!\!\!\!\! \W_{x,y} \log \W_{x,y} - \left(\frac{\mu_{2,n}}{\mu_{1,n}}-\log \mu_{1,n} \right) - \frac{1}{n}\!\!\sum_{x\in [n],y \in [m]}\!\!\!\!\! \W_{x,y} \log \sum_{k \in [n]} \W_{k,y} \right| \geq t} \\
&\hspace{6mm} = \mathds{P}\left[\left|\log m + \frac{1}{n} \sum_{x\in [n],y \in [m]} \W_{x,y} \log \W_{x,y} - \left(\frac{\mu_{2,n}}{\mu_{1,n}}-\log \mu_{1,n} \right) \right. \right. \\
&\hspace{20mm} \left. \left.- \frac{1}{n}\sum_{x\in [n],y \in [m]} \W_{x,y} \log \sum_{k \in [n]} \W_{k,y}+\log \frac{n}{m} \right| \geq t \right]\\
&\hspace{6mm} \leq \Prob{\left| \log m +\!\sum_{y \in [n]}\!\W_{\underline{x},y}\log\W_{\underline{x},y}\! - \!\left(\frac{\mu_{2,n}}{\mu_{1,n}}\!-\! \log \mu_{1,n} \right)\! -\! \frac{1}{n}\!\!\!\sum_{x\in [n],y \in [m]}\!\!\!\!\! \W_{x,y} \log \sum_{k \in [n]} \W_{k,y}\!+\!\log\frac{n}{m} \right|\geq t} \\
& \hspace{6mm} \leq 2f(\alpha_{t/4},m)+f(\tfrac{t}{4L},m)+f(\beta_{t/(2L)},n)+f(\beta_{t/(2L)},m),
\end{align*}
where the final inequality uses similar steps as done in the derivation of \eqref{eq:ddavv} together with Lemma~\ref{lem:basic} and Lemma~\ref{lem:LBiii}.
\end{proof}


\section*{Acknowledgments}
The authors thank Annina Bracher and Renato Renner for helpful discussions.
TS and JL were supported by the ETH grant (ETH-15 12-2).
DS acknowledges support by the Swiss National Science Foundation (through the National Centre of Competence in Research `Quantum Science and Technology' and grant No.~200020-135048) and by the European Research Council (grant No.~258932).

\bibliography{./bibtex/header,./bibtex/bibliofile}


\clearpage



\end{document}